\documentclass{article}

\usepackage{url}

\makeatletter
\renewcommand{\paragraph}[1]{%
  \par %
  \addvspace{1.0ex \@plus .2ex \@minus .2ex}%
  \noindent %
  {\normalfont\normalsize\scshape \MakeLowercase{#1}}\quad%
  \@afterindentfalse %
  \@afterheading     %
}
\makeatother

\usepackage{booktabs} %
\usepackage{tikz} %

\usepackage{adjustbox}
\usepackage{array}
\newcolumntype{L}[1]{>{\raggedright\let\newline\\\arraybackslash\hspace{0pt}}m{#1}}
\newcolumntype{C}[1]{>{\centering\let\newline\\\arraybackslash\hspace{0pt}}m{#1}}
\newcolumntype{R}[1]{>{\raggedleft\let\newline\\\arraybackslash\hspace{0pt}}m{#1}}

\usetikzlibrary{patterns,arrows,arrows.meta,calc,shapes,shadows,decorations.pathmorphing,decorations.pathreplacing,automata,shapes.multipart,positioning,shapes.geometric,fit,circuits,trees,shapes.gates.logic.US,fit,decorations.text,bending}

\tikzset{>=latex}
\usetikzlibrary{backgrounds} 

\usepackage{amsmath}
\usepackage{amsfonts,amssymb,pifont}
\usepackage[x11names]{xcolor}
\usepackage{colortbl,tabularx}
\usepackage{pgfplots}
\usepackage{algorithm}
\usepackage{algpseudocode}
\usepackage{stix}
\usepackage[letterpaper, total={5.5in, 9in}]{geometry}
\pgfplotsset{compat=1.18}

\usepackage{amsthm}

\newtheorem{theorem}{Theorem}

\newtheorem{example}{Example}
\newtheorem{definition}{Definition}
\newtheorem{problem}{Problem}

\newenvironment{continuance}[1]
  {\par\bigskip\noindent\textbf{Example #1. cont. }\itshape}
  {\par}

\newcommand{\TO}{TO}

\newcommand{\Act}{\ensuremath{\mathit{Act}}}
\newcommand{\act}{\ensuremath{a}}

\DeclareMathOperator{\supp}{supp}
\DeclareMathOperator{\successor}{succ}

\newcommand{\ObsSym}{{Z}}
\newcommand{\ObsFun}{{O}}
\newcommand{\obs}{\ensuremath{z}}
\newcommand{\pomdp}{\mathcal{M}}
\newcommand{\states}{\ensuremath{S}}
\newcommand{\init}{\ensuremath{I}}
\newcommand{\probmdp}{\mathcal{P}}
\newcommand{\rew}{\ensuremath{\mathcal{R}}}
\newcommand{\osched}{\ensuremath{\mathit{\pi}}}

\newcommand{\belief}{\mathfrak{b}}
\newcommand{\belsup}{B}
\newcommand{\belsups}{\mathcal{B}}
\newcommand{\avoid}{\textsl{AVOID}}
\newcommand{\reach}{\textsl{REACH}}

\newcommand{\psched}{\nu}
\newcommand{\estimator}{\sigma}

\newcommand{\probability}{\mathit{Pr}}

\definecolor{steve_red}{HTML}{FF0F19}
\definecolor{randomgray}{RGB}{128,128,128}
\definecolor{steve_green}{RGB}{1,68,33}
\definecolor{steve_yellow}{HTML}{FDE9CA}
\definecolor{steve_highlight}{HTML}{E99076}

\definecolor{headerblue}{RGB}{65, 81, 102}
\definecolor{myred}{RGB}{158, 37, 16}
\definecolor{mygreen}{RGB}{8, 158,95}
\definecolor{mydarkblue}{RGB}{8,59,158}
\definecolor{myyellow}{RGB}{158, 114,16}

\newcommand{\domain}[1]{\textsl{#1}}
\newcommand{\shield}{\mathsf{shield}}

\newcommand{\fillGrid}[3]{
		\fill[#1] (#2,#3) rectangle (#2+1,#3+1);
}

\makeatletter
\let\MYcaption\@makecaption
\makeatother

\usepackage{caption}
\usepackage{subcaption}

\makeatletter
\let\@makecaption\MYcaption
\makeatother

\usepackage{color,array}

\usepackage{enumitem}
\setlist[itemize]{leftmargin=*}

\title{Compositional~Shield Synthesis\\ for Safe Reinforcement Learning in~Partial Observability}
\author{Steven Carr${}^1$, Georgios Bakirtzis${}^2$, Ufuk Topcu${}^1$\\[.25em]${}^1$ The University of Texas at Austin\\${}^2$ LTCI, Télécom Paris, Institut Polytechnique de Paris}
\date{}

\begin{document}
\maketitle
\begin{abstract}Agents controlled by the output of reinforcement learning (RL) algorithms often transition to unsafe states, particularly in uncertain and partially observable environments. Partially observable Markov decision processes (POMDPs) provide a natural setting for studying such scenarios with limited sensing. \emph{Shields} filter undesirable actions to ensure safe RL by preserving safety requirements in the agents' policy. However, synthesizing holistic shields is computationally expensive in complex deployment scenarios. We propose the \emph{compositional} synthesis of shields by modeling safety requirements by parts, thereby improving scalability. In particular, problem formulations in the form of POMDPs using  RL algorithms illustrate that an RL agent equipped with the resulting compositional shielding, beyond being safe, converges to higher values of expected reward. By using subproblem formulations, we preserve and improve the ability of shielded agents to require fewer training episodes than unshielded agents, especially in sparse-reward settings. Concretely, we find that compositional shield synthesis  allows an RL agent to remain safe in environments two orders of magnitude larger than other state-of-the-art model-based approaches.
\end{abstract}

\section{INTRODUCTION}

Reinforcement learning (RL) has emerged as a promising approach for constructing controllers that can operate in uncertain and partially observable environments, often modeled as partially observable Markov decision processes (POMDPs)~\cite{conf/ijcai/Even-DarKM05}. However, the exploratory nature of RL algorithms can lead to the execution of unsafe actions, which is particularly concerning in safety-critical domains. Shielding, a technique that restricts the action space of the learning agent to a pre-computed set of safe actions, has been shown to improve both performance and safety compared to other approaches, such as reward shaping \cite{DBLP:conf/aaai/Carr0JT23,DBLP:conf/aaai/AlshiekhBEKNT18}.

Despite its advantages, shielding in POMDPs faces a significant scalability challenge due to the exponential growth of the belief space with respect to the number of states. In a POMDP, the agent maintains a belief distribution over the possible states, and the size of this belief space is exponential in the number of states. Consequently, the memory required to track the belief support grows exponentially, making centralized shield synthesis computationally intractable for large POMDPs \cite{chatterjee2009sensitivity,DBLP:conf/cav/JungesJS20}, and therefore, realistic environments. Moreover, the shield synthesis problem itself is known to be EXPTIME-complete in the number of agent and belief states, further exacerbating the scalability issue~\cite{chatterjee2015qualitative}.

To tackle this challenge, we introduce a compositional approach to shield synthesis that exploits the inherent structure present in many POMDP scenarios. By decomposing the POMDP into smaller, more manageable submodels and computing local shields for each submodel, the compositional synthesis algorithm increases the threshold of computable shields. In particular, the compositional shielding framework allows for the parallel computation of sub-shields and enables efficient belief tracking across the submodels, ultimately facilitating the application of shielding to much larger POMDPs than previously possible.

This modular safety reasoning enables the independent synthesis of local shields, each operating over a reduced state space. Recent works have shown where local shields can be computed and later composed to form a coherent global safety mechanism~\cite{melcer2024shield,neary2023verifiable}. Our proposed framework significantly expands the class of POMDPs for which shield synthesis is tractable, offering a scalable alternative to centralized methods \cite{DBLP:conf/cav/JungesJS20} while maintaining formal safety guarantees. This compositional view not only provides improvements through structure in reinforcement learning~\cite{neary2023verifiable,simpkins2019composable} but also aligns naturally with real-world systems, which are often modular by design~\cite{Rosolia2020unified}.

The main contributions of this work are threefold.
\begin{enumerate}
\item We develop a compositional shield synthesis approach that guarantees safety both during and after learning in large-scale POMDPs, addressing the scalability limitations of centralized shielding methods.
\item We propose a mechanism for information sharing and belief tracking across the sub-shields, ensuring the coherence of the compositional shielding framework.
\item We integrate our compositional shielding approach with state-of-the-art deep RL algorithms, demonstrating its ability to calculate shields in two orders of magnitude larger than current approaches.
\end{enumerate}

We design experiments that find that compositional shielding improves learning performance, sample efficiency, and scalability. Our method  outperforms centralized shielding and unshielded RL baselines, while enabling the synthesis of shields for POMDPs that are two orders of magnitude larger than those feasible for existing centralized methods. 

\section{RELATED WORK}

Safe RL is necessary for deploying agents in high-stakes environments such as autonomous driving and robotics, where failures can lead to significant consequences. Because RL relies on exploratory interaction to learn the effects of actions~\cite{DBLP:conf/iros/PetersS06}, it is prone to selecting unsafe behaviors that may transition the system into hazardous states. The challenge of \emph{unsafe exploration} motivates various strategies for reducing the likelihood of dangerous behavior~\cite{garcia2015comprehensive}. Safe RL can be broadly grouped into three categories:
\begin{enumerate}
    \item \textbf{Reward shaping:} Enhancing the reward function to reflect safety preferences~\cite{laud2003influence}, often based on human feedback or expert demonstration.
    \item \textbf{Constrained optimization:} Introducing safety through explicit cost functions or constraints~\cite{DBLP:conf/icml/MoldovanA12}, which penalize unsafe actions or states.
    \item \textbf{Runtime shielding:} Synthesizing a safety mechanism that restricts the agent’s actions during execution~\cite{DBLP:conf/aaai/AlshiekhBEKNT18}.
\end{enumerate}

This work falls within the third category and extends the theory and application of shielding in reinforcement learning. In contrast to modifying reward signals or imposing soft penalties, shielding provides formal guarantees by filtering unsafe actions at runtime~\cite{DBLP:conf/aaai/Carr0JT23}. While shielding has been well studied in fully observable Markov decision processes~\cite{fulton2018safe,DBLP:journals/corr/abs-1904-07189}, scaling such techniques to partially observable environments remains a challenge.

POMDPs introduce an additional layer of complexity for safe RL due to uncertainty in both action outcomes and observations~\cite{DBLP:journals/ai/KaelblingLC98}. Effective planning in POMDPs typically requires reasoning over histories or belief states, and recent work has made progress in scaling deep POMDP solvers~\cite{hausknecht2015deep,DBLP:journals/corr/MnihKSGAWR13}. However, these approaches do not incorporate safety constraints directly. Prior work on shielding for deep RL under partial observability~\cite{DBLP:conf/aaai/Carr0JT23} provides safety guarantees but encounters significant scalability limitations.
In this paper, we address these limitations by proposing a compositional approach to shield synthesis for POMDPs. By decomposing the state space and safety specification into submodels, we synthesize local shields independently and combine them to form a global shield. This approach significantly improves scalability while preserving safety guarantees.

Recent work has explored compositional shielding in multi-agent systems. Melcer et al.~\cite{melcer2024shield} propose a decentralized shielding mechanism for multi-agent systems without requiring inter-agent communication, focusing on empirical performance. Brorholt et al.~\cite{brorholt2024compositional} apply assume-guarantee reasoning to construct independent shields for agents under a global specification.  We view these works as complementary lines of research that pursue compositional safety using different methods.

In summary, this work contributes to the body of research on runtime shielding in safe RL by introducing a compositional synthesis method for POMDPs. Unlike prior work, our approach enables scalable shield synthesis while maintaining formal guarantees of safety, and is distinct from recent multi-agent compositional methods in scope and application.

\section{THEORETICAL PRELIMINARIES}
The standard model for sequential decision-making in the presence of uncertainty is the Markov decision process~(MDP). MDPs traditionally model uncertainty in the form of stochasticity in the transition dynamics~\cite{DBLP:books/daglib/BaierKatoen2008}.
POMDPs are an extension of MDPs with an additional source of uncertainty: incomplete information about the system state due to noisy or partial observations.
By addressing these sources of uncertainty, POMDPs are effective in modeling real-world applications like motion-planning with limited visibility \cite{thrun2005probabilistic} or sampling terrain with noisy sensors~\cite{smith2004heuristic}.

\subsection{Partially Observable Markov Decision Processes}

In this subsection, we introduce POMDPs \cite{DBLP:journals/ai/KaelblingLC98} and theoretical preliminaries associated with modeling uncertain environments to train agents operating with incomplete information.

\begin{definition}[POMDP]
A (discrete) POMDP is a tuple $\pomdp = (\states,\init,\Act, \ObsFun,\ObsSym, \probmdp, \rew)$ where:
\begin{itemize}
\item $\states$ is a finite set of states.
\item $\init \in \Delta(\states)$ is the initial state distribution, with $\init(s)$ the probability of starting in state $s\in\states$.
\item $\Act$ is a finite set of actions. 
\item $\ObsSym$ is a finite set of observations and $\ObsFun(\obs|s)$ is the probability of observing $\obs\in\ObsSym$ in state $s\in\states$.
\item $\probmdp(s'| s,\act)$ is the transition probability of moving to state $s'\in\states$ after taking action $\act\in\Act$ in state $s\in\states$.
\item $\rew: \states \times \Act \to \mathbb{R}$ is the reward function, with $\rew(s, \act)$ the reward for taking action $\act$ in state $s$.
\end{itemize}
\end{definition}

We denote the set of available actions in state $s$
by $\Act(s) \subseteq \Act$.
In this work, we consider POMDPs with \emph{dead-ends}, i.e., states from which the agent cannot obtain positive rewards in the future~\cite{DBLP:conf/uai/KolobovMW12}. We do not want the agent to end up in such dead-end states because once there, it cannot achieve any positive rewards.

By defining successor states,
 we are essentially helping the agent understand which states are 
reachable or safe to transition to from a given set of states, thereby avoiding decisions that would lead it to 
dead-ends.

\begin{definition}[Successor States]
For a POMDP $\pomdp$ and set of states $\states'\subseteq\states$, the successors $\successor(\states')$ is the set of states $s\in\states$ for which there exists an action $\act\in\Act(s)$ with $\probmdp(s'|\act,s) > 0$ for some $s'\in\states'$.
\end{definition}

\subsection{Beliefs}

A belief state $\belief \in \Delta(\states)$ is a probability distribution over states representing the agent's belief about the likelihood of being in each state given the history of observations and actions. The \emph{support} of a belief state $\belief$, denoted $\supp(\belief)$, is the set of states with non-zero probability under $\belief$, i.e., $\supp(\belief) = \{s \in \states \colon \belief(s) > 0\}$ \cite{DBLP:journals/lmcs/RaskinCDH07}.

\begin{definition}[Belief Support]
For a given belief $\belief$, a state $s$ is in the belief support $\belsup$ if and only if $s \in \supp(\belief)$. The set of all belief supports is denoted $\belsups = \{B \subseteq \states \colon B \neq \emptyset\}$, i.e., the set of non-empty subsets of $\states$.
\end{definition}

If the transition and observation probabilities in the POMDP are known, the agent can update its belief state using Bayes' rule after each action and observation. However, in many practical settings, these probabilities may be unknown. In such cases, we can still update the belief support $\belsup$ using only the graph structure of the POMDP, without relying on exact probability values.

\begin{definition}[Belief Support Update]
Given a current belief support $\belsup$, action $\act$, and observation $\obs$, the updated belief support $\belsup'$ is the set of states $s'$ such that there exists a state $s \in \belsup$ and a non-zero probability transition from $s$ to $s'$ under action $\act$, and $\obs$ can be observed in $s'$, meaning that updating the belief support amounts to building the set
$$\belsup' = \{s' \in \states \colon \exists s \in \belsup \text{ s.t. } \probmdp(s'|s,\act) > 0 \text{ and } \ObsFun(\obs|s') > 0\}.$$
\end{definition}

\begin{definition}[State Estimator]
A state estimator is a function $\estimator: (\ObsSym \times \Act)^* \times \ObsSym \to \belsups$ that takes as input the history of observations and actions and returns the current belief support. The estimator can be implemented by repeatedly updating the belief support using the graph structure of the POMDP, without requiring knowledge of the exact transition and observation probabilities.
\end{definition}

A policy $\pi$ for a POMDP is a function mapping histories of observations and actions to distributions over actions, i.e., $\pi: (\ObsSym \times \Act)^* \times \ObsSym \to \Delta(\Act)$. This allows the choice of action to depend on the history, which is necessary for the agent to act optimally given its uncertainty about the true state. The value of a policy $\pi$ from a belief state $\belief$ is the expected total discounted reward
$$V^{\pi}(b) = \mathbb{E}_{\pi,\belief}\left[\sum_{t=0}^{\infty} \gamma^t R(s_t, a_t)\right],$$
where $\gamma \in [0, 1)$ is the discount factor, and the expectation is taken with respect to the distribution over trajectories induced by $\pi$ starting from belief state $\belief$.

The optimal policy $\pi^*$ maximizes the value function for all belief states
$$\pi^* = \operatorname*{argmax}_{\pi} V^{\pi}(b), \quad \forall \belief \in \Delta(\states).$$
Computing the optimal policy for a POMDP is generally intractable~\cite{DBLP:journals/mor/PapadimitriouT87}, but various approximate solution methods exist, such as point-based value iteration \cite{pineau2003point} and online search methods \cite{ross2008online}.

\subsection{Shields}

In safety-critical settings, an agent must not only maximize rewards but also adhere to safety constraints. These constraints are captured using \emph{(qualitative) reach-avoid} specifications, a subclass of indefinite horizon properties~\cite{Put94}. Such specifications require the agent to \emph{always} avoid certain unsafe states, denoted by $$\avoid \subseteq \states,$$ and to reach certain states, denoted by $$\reach \subseteq \states,$$ \emph{almost surely}; i.e., with probability one over an indefinite horizon. We represent these constraints as $$\varphi = \langle \reach, \avoid \rangle.$$

The \emph{avoid} specification, denoted by $\varphi_A = \langle \avoid \rangle$, requires only the avoidance of unsafe states. The relation $\pomdp(\osched) \models \varphi_A$ indicates that the agent adheres to the specification $\varphi$ under policy $\osched$. An interesting attribute of these specification is that rewards are irrelevant for computing safe sets.

For a set of states $S' \subseteq S$ in a POMDP, $\probability_\belief^\osched(S')$ denotes the probability of reaching $S'$ from the belief $\belief$ using the policy $\osched$. A policy $\osched$ is \emph{winning} for the specification $\varphi$ from belief $\belief$ in the POMDP $\pomdp$ if and only if $\probability_\belief^\osched(\avoid) = 0$ and $\probability_\belief^\osched(\reach) = 1$. This means that the policy must reach $\avoid$ with probability zero and $\reach$ with probability one (almost surely) when $\belief$ is the initial state.

A belief $\belief$ is \emph{winning} for $\varphi$ in $\pomdp$ if there exists a winning policy from $\belief$. For multiple beliefs, we define \emph{winning regions} (also known as safe or controllable regions). A \emph{winning region} in a POMDP is a set of winning beliefs, meaning that from each belief within this region, there exists a winning policy.

The purpose of a shield is to prevent the agent from taking actions that would violate a (reach-avoid) specification. For \emph{avoid specifications}, the shield ensures that the agent does not enter predetermined avoid states or states from which it is impossible to prevent reaching an avoid state in the future. Consequently, a shield ensures that an agent stays within a winning region. To remain inside this region, the agent must select actions such that all successor states (from the current belief) also remain within the winning region.

A shield also prevents the agent from encountering dead ends in reach-avoid specifications. While a shield itself cannot force an agent to visit reach states, we can ensure that the agent eventually visits the reach state under mild assumptions~\cite{DBLP:conf/cav/JungesJS20}. Formally, we define a shield as a set of (winning) policies, often referred to as a \emph{permissive} policy~\cite{DBLP:journals/corr/DragerFK0U15,DBLP:conf/tacas/Junges0DTK16}.

To compute winning policies, beliefs, and regions for qualitative reach-avoid properties, we only need to consider the finite set of belief supports~\cite{DBLP:journals/lmcs/RaskinCDH07}. Constructing a finite, albeit exponential, belief-support model that suitably abstracts the belief MDP is a prerequisite to shielding. Denote the set of all belief supports by $\belsup$. We then directly define policies on the belief support as $\osched^{\belief} \colon \belsup \rightarrow \Act$. This \emph{deterministic} policy selects a unique action for each belief support $\supp(\belief)$.

\begin{definition}[Shield] \label{def:shield}
	A permissive policy for a POMDP $\pomdp$ is given by the function $\psched\colon \belief \rightarrow 2^\Act$. 
	A policy $\osched$ is \emph{admissible} for $\psched$ if for all beliefs $\belief$ it holds that $\osched(\belief)\in \psched(\belief)$.
	A permissive policy is a \emph{$\varphi$-shield for $\pomdp$} if all its admissible policies are winning.
\end{definition}

We present the general problem for enforcing an agent to remain safe in POMDPs~\cite{DBLP:conf/cav/JungesJS20}.
\begin{problem}[Ensured safety in POMDPs]
	Given a POMDP~$\pomdp$, a safety constraint $\varphi$, and a sequence of policies employed by an agent $\osched_1,\dots,\osched_n$, ensure that for all policies~$\osched_i$ it holds that $\pomdp(\osched_i)\models\varphi$ with $1\leq i\leq n$. %
\end{problem}

Precise, accurate models of uncertain systems are expensive to obtain, requiring the collection of large amounts of data or expert domain knowledge~\cite{baier201910}. However, the system's limitations that give rise to those uncertainties are often known. This tension motivates the use of a partial model. We assume the agent only has access to a partial model $\pomdp' = (\states, \init, \Act, \ObsFun, \ObsSym, \probmdp')$ where the transition model $\probmdp'$ yields unknown but positive probabilities. Essentially, $\probmdp'$ defines a set of possible transitions. 

A POMDP $\pomdp = (\states, \init, \Act, \ObsFun, \ObsSym, \probmdp)$ and a partial model $\pomdp' = (\states, \init, \Act, \ObsFun, \ObsSym, \probmdp')$ have \emph{coinciding transitions} if and only if for all states $s, s' \in \states$ and actions $\act \in \Act$, $\probmdp(s' \mid s, \act) > 0$. The partial model defines exactly the graph of the original POMDP. Similarly, $\pomdp'$ \emph{overapproximates the transition model of $\pomdp$} if for all states $s, s' \in \states$ and actions $\act \in \Act$, $\probmdp(s' \mid s, \act) > 0$ if $\probmdp'(s' \mid s, \act) > 0$. The original POMDP has no transitions that are not present in the partial model. 

The provable guarantees produced by a shield depend on the partiality of the model. Knowing the exact set of transitions with arbitrary positive probability for a POMDP is sufficient to compute a $\varphi$-shield~\cite{DBLP:conf/aaai/Carr0JT23}.

\section{COMPOSITIONAL SHIELD SYNTHESIS}

We first outline decomposion of a POMDP $\pomdp$ into a set of submodels $\{ \pomdp_1, \dots, \pomdp_n \}$ and describe properties that make a composition admissible. Then we describe shielding and the guarantees computed on the submodels.  We present a working example in the domain \domain{Obstacle(8)} (example~\ref{ex:grid}):

\begin{example}\label{ex:grid}
	Consider the environment \domain{Obstacle(8)} (Figure~\ref{fig:Abstraction_Sets}), which is a single-agent gridworld with four initial locations and five ``dead zones'' from which the agent cannot exit.
	The agent has three observations $\left\{\text{white},\text{red},\text{green}\right\}$ where white is an unlabelled state, red is a dead zone, and green is the target state.
	At each state, the agent may select one of four actions $\left\{\textrm{left},\textrm{right},\textrm{up},\textrm{down}\right\}$, upon selecting the agent, the agent may move one or two spaces depending on the probability of slip $p=0.1$.
    In such an environment, there is a natural segmentation into smaller grids as quadrants (dashed lines of Figure~\ref{fig:Abstraction_Sets}).
 \end{example}

\begin{figure*}[t!]
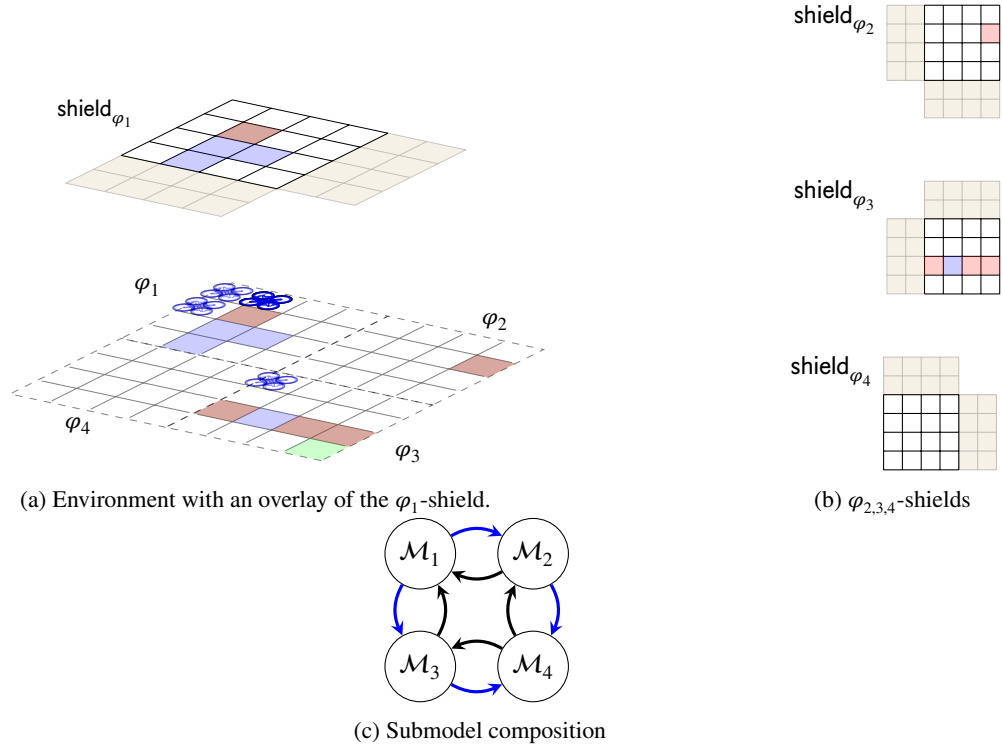

    \centering
    \subcaptionbox{Environment with an overlay of the $\varphi_1$-shield.}[8cm]{
    	\resizebox{8.0cm}{!}{%
        \newcommand{\translatepoint}[1]%
{   \coordinate (mytranslation) at (#1);
}

\newcommand{\gridThreeD}[2]{
	\begin{scope}[yshift=#1,yslant=0.5,xslant=-1.4,scale=0.4]
		\draw [step=1,opacity=#2] (0.1,0.1) grid (3.9,3.9);
		\draw [dashed,opacity=#2] (0,0) rectangle (4,4);
		\draw [step=1,opacity=#2] (4.1,4.1) grid (7.9,7.9);
		\draw [dashed,opacity=#2] (4,4) rectangle (8,8);
		\draw [step=1,opacity=#2] (0.1,4.1) grid (3.9,7.9);
		\draw [dashed,opacity=#2] (0,4) rectangle (4,8);
		\draw [step=1,opacity=#2] (4.1,0.1) grid (7.9,3.9);
		\draw [dashed,opacity=#2] (4,0) rectangle (8,4);
		\fillGrid{myred,opacity=#2-0.1}{1}{1}
		\fillGrid{myred,opacity=#2-0.1}{1}{0}
		\fillGrid{myred,opacity=#2-0.1}{1}{3}
		\fillGrid{myred,opacity=#2-0.1}{6}{0}
		\fillGrid{myred,opacity=#2-0.1}{6}{6}
		\fillGrid{blue,opacity=#2-0.3}{5}{6}
		\fillGrid{blue,opacity=#2-0.3}{5}{5}
		\fillGrid{blue,opacity=#2-0.3}{4}{6}
		\fillGrid{blue,opacity=#2-0.3}{1}{2}
		\fillGrid{green,opacity=#2-0.3}{0}{0}
%

    %

	\end{scope}
}

\newcommand{\UAVLayer}[2]{
    \foreach[evaluate={\n=0.85-abs(0.02-\i)}] \i in {0,0.005,0.01,0.015,0.02,0.025,0.03,0.035,0.04}
    {
	\begin{scope}[yshift=#1,yslant=0.5,xslant=-1.4,scale=0.4]
	
        \UAVopaq{1}{2}{\n}{#2}{0.05}
        \UAVopaq{5}{5}{\n}{#2}{0.2}
        
        \UAVopaq{5}{6}{\n}{#2}{0.05}
        \UAVopaq{4}{6}{\n}{#2}{0.05}
        
	\end{scope}
	}
}

\newcommand{\drawLinewithBG}[2]
{
	\draw[white]  (#1) -- (#2);
	\draw[black,very thick,dashed] (#1) -- (#2);
}

\tikzset{%
	thick arrow/.style={
		-{Triangle[angle=120:1pt 1]},
		line width=1.5cm, 
		draw=gray
	},
	arrow label/.style={
		text=white,
		font=\sf,
		align=center
	},
	set mark/.style={
		insert path={
			node [midway, arrow label, node contents=#1]
		}
	}
}

	\input{figures/UAV_opaque}
	\input{figures/UAV}

\begin{tikzpicture}

\gridThreeD{0}{0.5}
\UAVLayer{20}{blue}

\begin{scope}[yshift=75,yslant=0.5,xslant=-1.4,scale=0.4]
		\draw [step=1] (4,4) grid (8,8);
		\fill[myyellow,opacity=0.1] (2,4) rectangle (4,8);
		\draw [step=1,opacity=0.2] (2,4) grid (4,8);
		
		\fill[myyellow,opacity=0.1] (4,2) rectangle (8,4);
		\draw [step=1,opacity=0.2] (4,2) grid (8,4);
		
		\fillGrid{myred,opacity=0.4}{6}{6}
		\fillGrid{blue,opacity=0.2}{5}{6}
		\fillGrid{blue,opacity=0.2}{5}{5}
		\fillGrid{blue,opacity=0.2}{4}{6}
%

    %

	\end{scope}

\node[fill=none,draw=none] at (-2.5,2.5) {$\varphi_1$};
\node[fill=none,draw=none] at (-3.25,5.0) {$\shield_{\varphi_1}$};
\node[fill=none,draw=none] at (2.5,2.0) {$\varphi_2$};
\node[fill=none,draw=none] at (1.25,0.1) {$\varphi_3$};
\node[fill=none,draw=none] at (-3.5,0.5) {$\varphi_4$};
\end{tikzpicture}

    	}
    }\hfill
    \subcaptionbox{$\varphi_{2,3,4}$-shields\label{fig:Abstraction_Sets_sub}}[3cm]{
        \centering
            \begin{tikzpicture}[scale=0.25]

\draw [step=1] (0,0) grid (4,4);
\fill[myyellow,opacity=0.1] (-2,0) rectangle (0,4);
\draw [step=1,opacity=0.2] (-2,0) grid (0,4);
		
\fill[myyellow,opacity=0.1] (0,-2) rectangle (4,0);
\draw [step=1,opacity=0.2] (0,-2) grid (4,0);

\fillGrid{red,opacity=0.2}{3}{2}
\node[draw=none,fill=none] at (-4.75,3.25) {$\shield_{\varphi_2}$};

\end{tikzpicture}
            \vspace{0.25cm}
            \begin{tikzpicture}[scale=0.25]

\draw [step=1] (0,0) grid (4,4);
\fill[myyellow,opacity=0.1] (0,4) rectangle (-2,0);
\draw [step=1,opacity=0.2] (0,4) grid (-2,0);
		
\fill[myyellow,opacity=0.1] (4,4) rectangle (0,6);
\draw [step=1,opacity=0.2] (4,4) grid (0,6);
\fillGrid{red,opacity=0.2}{0}{1}
\fillGrid{blue,opacity=0.2}{1}{1}
\fillGrid{red,opacity=0.2}{2}{1}
\fillGrid{red,opacity=0.2}{3}{1}

\node[draw=none,fill=none] at (-4.75,5.25) {$\shield_{\varphi_3}$};
\end{tikzpicture}
            \vspace{0.25cm}
            \begin{tikzpicture}[scale=0.25]

\draw [step=1] (0,0) grid (4,4);
\fill[myyellow,opacity=0.1] (0,4) rectangle (4,6);
\draw [step=1,opacity=0.2] (0,4) grid (4,6);
		
\fill[myyellow,opacity=0.1] (4,0) rectangle (6,4);
\draw [step=1,opacity=0.2] (4,0) grid (6,4);
\node[draw=none,fill=none] at (-2.75,5.25) {$\shield_{\varphi_4}$};
\end{tikzpicture}
    }\hfill
    \subcaptionbox{Submodel composition\label{fig:Abstraction_Sets_super}}[4cm]{
        \begin{tikzpicture}[scale=1.5, state/.append style={minimum size=2mm,inner sep=3pt},>=stealth,
        bobbel/.style={minimum size=2.8mm,inner sep=0pt,fill=black,circle},bobbelblack/.style={minimum size=2.0mm,inner sep=0pt,fill=black,circle}]

\node[state,fill=none] (s0) at (0,1) {$\pomdp_1$};
\node[state,fill=none] (s1) at (1,1) {$\pomdp_2$};
\node[state,fill=none] (s2) at (0,0) {$\pomdp_3$};
\node[state,fill=none] (s3) at (1,0) {$\pomdp_4$};

\draw[->,very thick,blue] (s0) to [bend left] (s1);
\draw[->,very thick,black] (s1) to [bend left] (s0);
\draw[->,very thick,blue] (s1) to [bend left] (s3);
\draw[->,very thick,black] (s3) to [bend left]  (s1);

\draw[->,very thick,blue] (s0) to [bend right] (s2);
\draw[->,very thick,black] (s2) to [bend right]  (s0);

\draw[->,very thick,blue] (s2) to [bend right]  (s3);
\draw[->,very thick,black] (s3) to [bend right]  (s2);

\end{tikzpicture}
    }\hfill   
    
    \caption{\domain{Obstacle(8)} equipped with a compositional shield, comprising four sub-shields that cover the entire state space using distinct sets of states and form connections using the set of potential successor states (highlighted in yellow). The reach-avoid specification $\varphi$ can be composed as $\bigcup_{i \in \lbrace 1 \dots 4 \rbrace} \varphi_i$, where each $\varphi_i$ represents the reach-avoid specification applied to its distinct set of states. The connected submodels form a composition graph, which helps us to define winning regions when the reach states are outside of the submodel's set of states (section~\ref{sec:reach-avoid-shield}).}
    \label{fig:Abstraction_Sets}
\end{figure*}

\subsection{POMDP (De)composition}

When decomposing the POMDP model $\pomdp$ into submodels $\{ \pomdp_1,\dots, \pomdp_n \}$, there are multiple approaches to achieving a compositional abstraction. Some  examples include:
\begin{itemize}
    \item distinguishing the location of an agent from the location of a moving obstacle,
    \item splitting a large area into smaller areas, and
    \item separating the tracking of the fuel reserves from the avoiding of collisions.
\end{itemize}
While each corresponds to different features in the problem formulation, all decompositions have the following in common: The subproblems are all performed on smaller state spaces than their precursor models. After decomposing and computing the behavior induced within the submodels, we must preserve some properties to compose.

\begin{definition}[State-based POMDP decomposition]\label{def:decom}
Let \( \mathcal{M} = (\states, \init, \Act, \ObsFun, \ObsSym, \probmdp, \rew) \) be a POMDP.
Given a subset of states \( \states_i \subseteq \states \), we define the corresponding \emph{submodel POMDP} as
\[
\mathcal{M}_i = (\hat{\states}_i, \init_i, \Act, \ObsFun_i, \ObsSym, \probmdp_i, \rew_i),
\]
where:
\begin{itemize}
    \item The submodel state space is \( \hat{\states}_i = \states_i \cup \successor(\states_i) \), where
    \[
    \successor(\states_i) = \left\{ s' \in \states \mid \exists s \in \states_i,\ a \in \Act,\ \probmdp(s' \mid s, a) > 0 \right\}.
    \]

    \item The transition function \( \probmdp_i \) is defined as:
    \[
    \probmdp_i(s' \mid s, a) =
    \begin{cases}
        \probmdp(s' \mid s, a) & \text{if } s \in \states_i \text{ and } s' \in \hat{\states}_i, \\
        1 & \text{if } s \in \hat{\states}_i \setminus \states_i \text{ and } s' = s, \\
        0 & \text{otherwise}.
    \end{cases}
    \]
    That is, states in \( \hat{\states}_i \setminus \states_i \) are treated as \emph{absorbing} states with self-loops for all actions.
    We emphasize that while interface states \( \hat{\states}_i \setminus \states_i \) are included to model successors and enable belief tracking, they are not controllable in \( \mathcal{M}_i \) and are never initial. These “absorbing” interface states are a conservative overapproximation to handle belief updates at submodel boundaries and capture the condition when the agent leaves the domain of that submodel. They serve solely to close the dynamics of a given submodel upon exit (guaranteed local safety) and preserve successor structure without underestimating risk transiting from one submodel to another.

    \item The observation function \( \ObsFun_i \) is the restriction of \( \ObsFun \) to \( \hat{\states}_i \), i.e.,
    \[
    \ObsFun_i(z \mid s) = \ObsFun(z \mid s) \quad \text{for all } s \in \hat{\states}_i,\ z \in \ObsSym.
    \]

    \item The reward function \( \rew_i \) is the restriction of \( \rew \) to core states:
    \[
    \rew_i(s, a) =
    \begin{cases}
        \rew(s, a) & \text{if } s \in \states_i, \\
        0 & \text{otherwise}.
    \end{cases}
    \]

    \item The initial distribution \( \init_i \) must satisfy:
    \[
    \supp(\init) \cap \states_i \subseteq \supp(\init_i) \, \text{and} \, \init_i(s) = 0 \text{ for all } s \in \hat{\states}_i \setminus \states_i.
    \]
\end{itemize}

\end{definition}

\begin{continuance}{\ref{ex:grid}}
    Returning to Example~\ref{ex:grid}, for an agent that may be at the top-right corner (state $s \in \states_1$) of submodel $\pomdp_1$. The action $\textrm{right}$ would lead to $s'\in   \hat{\states}_i \setminus \states_i$. It is important that this successor $s$ is included in the states modelled by the submodel $\pomdp_1$ as we want to ensure a safe transition out of the submodel. By then making this state $s'$ absorbing, we reflect that $\pomdp_1$ does not have control over what happens after the agent $s'$, which is the responsibility of submodel $\pomdp_2$. This design ensures that each submodel operates only within its assigned domain and handles transitions conservatively at the edges. The absorbing behavior prevents circular dependencies between submodels and makes shield computation tractable by enforcing a clear boundary.
\end{continuance}

We describe the conditions on the initial distribution $\init_i$ (definition~\ref{def:admissible}) but at a minimum we preserve any initial states in $\init(s)$ to the submodel initial states $I_i(s)$ by restricting the domain to $\states_i$.

\begin{definition}[Admissible composition]\label{def:admissible}
A set of submodels \( \{ \pomdp_1, \dots, \pomdp_n \} \) is an \emph{admissible composition} of a POMDP \( \pomdp = (\states, \init, \Act, \ObsFun, \ObsSym, \probmdp, \rew) \) if the following conditions hold:
\begin{enumerate}
    \item[P1] State cover: Every state in \( \states \) is included in at least one submodel, i.e.,
    \[
    \states = \bigcup_{i=1}^n \states_i,
    \]
    where \( \states_i \) is the set of core states in submodel \( \pomdp_i \).

    \item[P2] Initialization consistency: For each submodel \( \pomdp_i \) with initial distribution \( \init_i \), the support of the global initial distribution \( \init \) restricted to \( \states_i \) must be contained in the support of \( \init_i \), i.e.,
    \[
    \supp(\init) \cap \states_i \subseteq \supp(\init_i).
    \]
\end{enumerate}
\end{definition}
To ensure that all transitions coincide we must include the successor states $\successor(\states_i) \subseteq \states$. However, it is important to distinguish that while submodel $\pomdp_i$ may contain states $s\in \successor(\states_i ) \backslash \states_i$ that are $s\notin \states_i $ we do not use the submodel $\pomdp_i$ to model the behavior in these states. We use these states to describe endpoints for the model $\pomdp_i$.
In this work, unless otherwise described, we refer to set of states inside of model $\pomdp_i$ as $\states_i$ rather the $\hat{\states_i}$ that they contain.

A set of submodels may have overlapping states. A submodel may, also, include more states in their initial distribution $I_i$ compared to the original model $\pomdp$, i.e.,~$|\supp(I(s))| < |\supp(I_i(s))|$, to allow for an overapproximated composition.

\subsection{Compositional Avoid Shield}

To form a composition, we consider the features of the problem and attempt to capitalize on symmetries. In Example~\ref{ex:grid}, a standard gridworld,
While state estimation (tracking of the belief support $\belsup$, see right-hand side of Figure~\ref{fig:synthesis_outline}) may still be conducted on the full model $\pomdp_i$, we require an interface for estimating the belief support inside the submodels to compute a shield.

\begin{definition}[Submodel state estimator]\label{def:submodel-estimator}
Let \( \pomdp_i = (\hat{\states}_i, \init_i, \Act, \ObsFun_i, \ObsSym, \probmdp_i, \rew_i) \) be a submodel defined over core state set \( \states_i \subseteq \states \) and extended state set \( \hat{\states}_i \supseteq \states_i \).

Let \( \belief_i \colon \ObsSym \times \Act \times \ObsSym \rightarrow \Delta(\states_i) \) denote the submodel belief distribution over \( \states_i \), and let the corresponding \emph{belief support} be defined as:
\[
\belsup_i = \supp(\belief_i) \subseteq \states_i.
\]
Let \( \belsups_i \) denote the set of all possible belief supports over \( \states_i \).

Then, a \emph{submodel state estimator} is a function
\[
\estimator_i \colon (\ObsSym \times \Act)^* \rightarrow \belsups_i
\]
that, given a sequence of observation–action pairs, computes the updated belief support \( \belsup_i \) over \( \states_i \), using the transition relation \( \probmdp_i \) and observation function \( \ObsFun_i \).

Specifically, \( \estimator_i \) maintains a sound overapproximation of reachable states in \( \states_i \) given the agent’s history \( h \in (\ObsSym \times \Act)^* \), based on the POMDP graph structure of the submodel.
\end{definition}

\begin{figure*}
\centering
		\newcommand*{\connectorH}[4][]{
  \draw[#1] (#3) -| ($(#3) !#2! (#4)$) |- (#4);
}
\newcommand*{\connectorV}[4][]{
  \draw[#1] (#3) |- ($(#3)+(0,#2)$) -| (#4);
}

\newcommand{\splitdistance}{1.65cm}
\begin{tikzpicture}
    \node[draw, minimum width=3.0cm,text width=3.0cm,rounded corners=0.05cm,align=center] (pomdp) { Partial model $\pomdp'$ \&\\[-0.2em] Specification $\varphi$};
    \node[below=1.0cm of pomdp,draw, minimum width=2.75cm,text width=2.75cm,rounded corners=0.05cm,align=center] (feature) { Feature splitting, model extension};
    \node[below=1.5cm of feature,xshift=-3*\splitdistance,,draw, minimum width=2.85cm,text width=2.85cm,inner sep=0.05cm,rounded corners=0.05cm,align=center] (subone) { POMDP $\pomdp_1'$ \&\\[-0.2em] Specification $\varphi_1$};
    \node[below=1.5cm of feature,xshift=-1*\splitdistance,draw,minimum width=2.85cm,text width=2.85cm,inner sep=0.05cm,rounded corners=0.05cm,align=center] (subtwo) { POMDP $\pomdp_2'$ \& \\[-0.2em]Specification $\varphi_2$};
     \node[below=1.5cm of feature,xshift=2*\splitdistance,draw,minimum width=2.75cm,text width=2.75cm,inner sep=0.05cm,rounded corners=0.05cm,align=center] (subn) { POMDP $\pomdp_{n}'$ \&\\[-0.2em] Specification $\varphi_{n}$};
     \node[below=1.5cm of feature,xshift=0.5*\splitdistance](location){\Huge$\cdots$};
    \node[below=0.55cm of subone,rounded corners=0.05cm,align=center,text width=1.7cm] (sched1) {Permissive\\[-0.2em] policy: $\psched_1$};
    \node[below=0.55cm of subtwo,rounded corners=0.05cm,align=center,text width=1.7cm] (sched2) {Permissive\\[-0.2em] policy: $\psched_2$};
    \node[below=0.55cm of subn,rounded corners=0.05cm,align=center,text width=1.85cm] (schedn) {Permissive\\[-0.2em] policy: $\psched_{n}$};
    \node[draw, fit=(sched1)(schedn), minimum height=1.5cm, inner sep=2pt,dashed,rounded corners=0.05cm,label=192:{Sub-shields}] (outline) {};
    \node[right=2*\splitdistance of pomdp,rounded corners=0.05cm,align=center] (stateest) {State estimator: $\estimator$};
    \node[above=0.5cm of stateest,rounded corners=0.05cm,text width=3.2cm] (obsin) {Observation $\obs\in \ObsSym$\\[-0.2em] \& Action $a\in\Act$};
    
    \node[below=6cm of feature,text width=3cm,minimum width=2.75cm,text width=2.75cm,rounded corners=0.05cm,align=center,draw](allowact){Allowed actions\\[-0.2em] $\Act'\subseteq \Act$};

    \draw[thick,->] (pomdp) --node[right,text width=3.2cm]{Create sub-spaces} (feature);
    \draw[thick,->] (feature.south) -- (subone.north);
    \draw[thick,->] (feature.south) -- (subtwo.north);
    \draw[thick,->] (feature.south) -- (subn.north);
    \draw[thick,->] (subone) -- (sched1.north);
    \draw[thick,->] (subtwo) -- (sched2.north);
    \draw[thick,->] (subn) -- (schedn.north);
    
    \draw[thick,->] (pomdp) --node[above,text width=2.75cm]{} (stateest.west);
    \draw[thick,->] (obsin) --node[above,text width=2.75cm]{} (stateest.north);
    \draw[very thick,->] (pomdp)++(0,1.5) --node[above,text width=2.75cm]{} (pomdp.north);
    \draw[thick,->] (stateest) |-node[left,text width=2.25cm,yshift=3.5cm,xshift=.5cm]{Belief \\[-0.2em] support $\belsup$} (outline.east);
    \draw[thick,->] (sched1) -- node[left, xshift=-.25cm]{$\Act_1'$} (allowact);
    \draw[thick,->] (sched2) -- node[left]{$\Act_2'$} (allowact);
    \draw[thick,->] (schedn) -- node[left, xshift=-.25cm]{$\Act_{n}'$} (allowact);
 	\draw[very thick,->] (allowact) -- ($(allowact.east)+(1.0,0)$);

\end{tikzpicture}
	\caption{Feature-based compositional shield synthesis.}
	\label{fig:synthesis_outline}
\end{figure*}
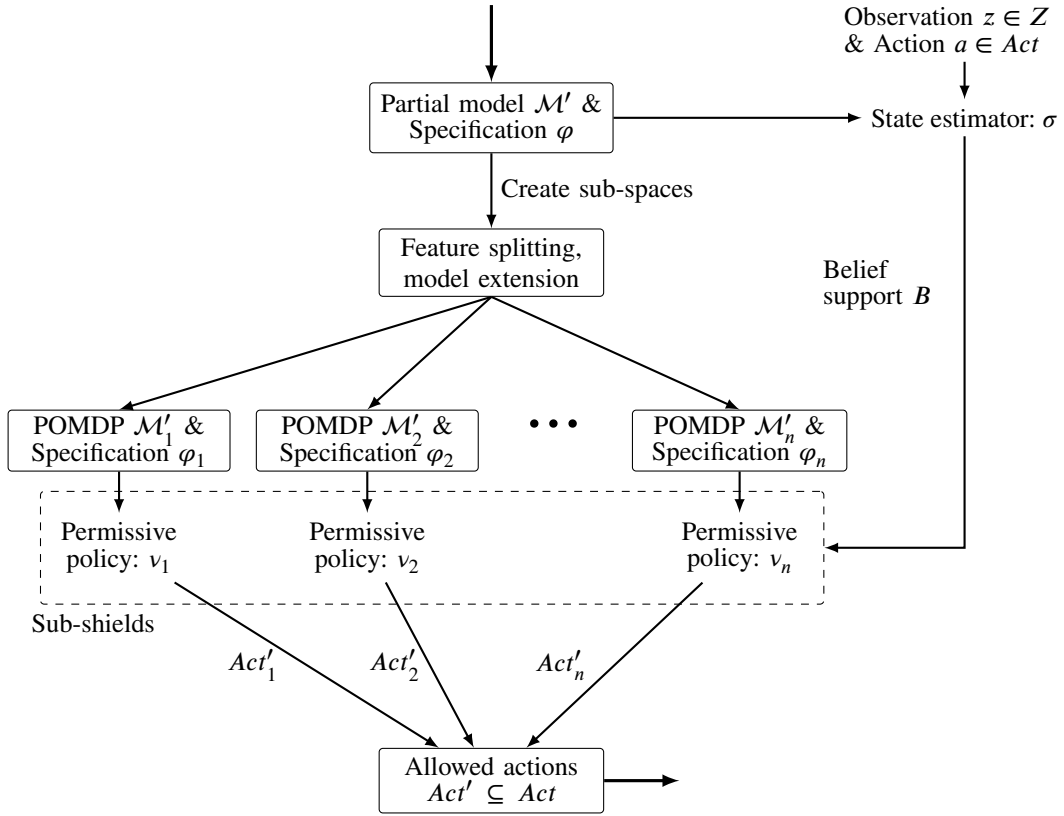

\paragraph{Consistency of belief support.}
We assume that the agent interacts with the environment by producing a trace of observations and actions. Let \( h = (z_0, a_0, z_1, a_1, \dots, z_t) \in (\ObsSym \times \Act)^* \) denote such a trace. The \emph{belief support} \( \belsup \subseteq \states \) of the global POMDP \( \pomdp \) at time \( t \) is computed using a \emph{state estimator} \( \estimator \colon (\ObsSym \times \Act)^* \to 2^{\states} \), which maps the agent’s trace to a set of states consistent with the observations and transitions in \( \pomdp \). That is,
\[
\belsup = \estimator(h).
\]

For each submodel \( \pomdp_i \), we define a corresponding \emph{submodel state estimator} \[\estimator_i \colon (\ObsSym \times \Act)^* \to 2^{\states_i}, \] which tracks the belief support \( \belsup_i = \estimator_i(h) \subseteq \states_i \) based only on the submodel's transition graph and observation function.

The \emph{compositional belief support} is then given by the union of the submodel belief supports:
\[
\hat{\belsup} = \bigcup_{i=1}^n \belsup_i.
\]

We say that the belief support is \emph{consistent under composition} if the global estimator and the union of submodel estimators yield the same result, i.e.,
\[
\estimator(h) = \bigcup_{i=1}^n \estimator_i(h) \quad \text{for all } h \in (\ObsSym \times \Act)^*.
\]

Under this consistency condition, belief updates also align: if \( a \in \Act \) and \( z \in \ObsSym \) are the latest action and observation, then
\[
\belsup' = \estimator(h \cdot (a, z)) = \bigcup_{i=1}^n \estimator_i(h \cdot (a, z)) = \hat{\belsup}'.
\]

There are two sufficient conditions for consistency, both of which come by construction in Definitions~\ref{def:decom}-\ref{def:submodel-estimator}: 1.) a shared trace input, where all submodels process a common trace h and 2.) consistent graph interfaces, i.e. the transition and observation functions used by $\estimator_i$ are faithful restrictions or projections of the global model. 
Importantly, belief consistency does not require exact agreement on belief distributions across submodels.  Because all estimators process the same trace $h$, and use deterministic updates based on the known  graph-structure of their respective submodels, the union of their outputs naturally conserves the reachable set as computed globally.

\paragraph{Compositional avoid procedure.}
Given a global avoid specification \( \varphi' = \langle \avoid \rangle \), each submodel \( \pomdp_i \) defined over core states \( \states_i \) and extended state space \( \hat{\states}_i \) inherits a local avoid specification \( \varphi'_i = \langle \avoid_i \rangle \), where
\[
\avoid_i = \avoid \cap \hat{\states}_i.
\]
That is, all avoid states from the global specification that are present in submodel \( \pomdp_i \)---whether in its core or interface states---must also be treated as avoid in the submodel specification. 
This ensures that any shield synthesized for \( \pomdp_i \) prevents transitions into states that violate the global avoid specification, maintaining soundness under composition.

\begin{theorem}[Compositional $\avoid$-shield]\label{thm:comp-avoid-shield}
Let $\pomdp'=(\states,\init,\Act,\ObsFun,\ObsSym,\probmdp',\rew)$ be a POMDP with avoid set $\avoid\subseteq\states$.  
Let $\{\pomdp'_1,\dots,\pomdp'_n\}$ be submodels forming an admissible composition of $\pomdp'$ (Definition~\ref{def:admissible}).  
For each $i \in \{1,\dots,n\}$, define
$\avoid_i = \avoid \cap \hat{\states}_i$,
$\varphi'_i = \langle \avoid_i \rangle$,
and let $\psched_i$ be a $\varphi'_i$-shield for $\pomdp'_i$ (i.e., no trace allowed by $\psched_i$ in $\pomdp'_i$ can visit a state in $\avoid_i$).  
Define the global shield
\[
\psched(b) = \bigcap_{i=1}^n \psched_i\bigl(\belsup_i(b)\bigr),
\]
where $\belsup_i(b)$ is the submodel-$i$ belief support computed by the submodel state estimator $\estimator_i$ (Definition~\ref{def:submodel-estimator}).  
Then $\psched$ is a $\varphi'$-shield for $\pomdp'$: no trace allowed by $\psched$ in $\pomdp'$ can ever visit a state in $\avoid$.
\end{theorem}

\begin{proof}
We prove the claim by induction on the length $t$ of the observation–action trace
\[
h_t = (z_0, a_0, z_1, \dots, a_{t-1}, z_t).
\]

\paragraph{Base case \((t = 0)\)}
The agent starts in some state $s_0 \sim \init$.  
By admissible composition (Definition~\ref{def:admissible}, property P2), every $s_0 \in \supp(\init)$ lies in at least one submodel $\pomdp'_i$ and hence in $\states_i \subseteq \hat{\states}_i$.  
Since $\avoid_i = \avoid \cap \hat{\states}_i$, we know $s_0 \notin \avoid_i$ and hence $s_0 \notin \avoid$.  
Thus the base case holds.

\paragraph{Inductive step}
Assume that after trace $h_t$ of length $t$, the agent's belief support $\belsup(h_t) \subseteq \states$ satisfies $\belsup(h_t) \cap \avoid = \emptyset$.

Let $a_t$ be an action allowed by the shield $\psched$, i.e.,
\[
a_t \in \psched\bigl(\belsup(h_t)\bigr) = \bigcap_{i=1}^n \psched_i\bigl(\belsup_i(h_t)\bigr),
\]
where $\belsup_i(h_t) = \estimator_i(h_t)$ is the belief support tracked within submodel $\pomdp'_i$ using its own estimator.

Since $a_t$ is allowed by every submodel shield on its local belief support, and each $\psched_i$ is a $\varphi'_i$-shield, it follows that for all $s \in \belsup_i(h_t)$, no successor state
\[
s' \in \successor(s, a_t) \cap \avoid_i
\]
is reachable in $\pomdp'_i$.

By admissible composition (Definition~\ref{def:admissible}, property P1) and the construction of each submodel (Definition~7), every successor state 
\[
s' \in \successor(\belsup(h_t), a_t)
\]
in the global model $\pomdp'$ lies in at least one $\hat{\states}_i$ and hence is subject to some local shield that prevents entry into $\avoid_i$.  
Because $\avoid_i = \avoid \cap \hat{\states}_i$, we conclude $s' \notin \avoid$.  
Thus the updated belief support $\belsup(h_{t+1}) = \estimator(h_{t+1})$ satisfies
\[
\belsup(h_{t+1}) \cap \avoid = \emptyset.
\]
By induction, no trace permitted by the global shield $\psched$ can reach a state in $\avoid$.

\paragraph{Only if}
If the set $\{\pomdp'_1, \dots, \pomdp'_n\}$ does not form an admissible composition of $\pomdp'$, then either some state $s \in \states$ is not covered by any $\states_i$ (violating~P1),~or
some $s \in \supp(\init)$ is not included in any $\supp(\init_i)$~(violating P2),
in which case no set of local shields can guarantee global avoidance of $s \in \avoid$; i.e.,  shielding is unsound or incomplete.
\end{proof}

\vspace{.5em}
\subsection{Compositional Reach-Avoid Shield}
\label{sec:reach-avoid-shield}
While the avoid shield from the previous section restricts actions that could enter $\avoid$ states, it is necessary to construct shields that also allow reaching the target.

\begin{definition}[Winning submodel]
Let $\pomdp = (\states, \init, \Act, \ObsFun, \ObsSym, \probmdp, \rew)$ be a POMDP with a reach set $\reach \subseteq \states$. A submodel $\pomdp_i$ with core state set $\states_i$ is called a \emph{winning submodel} if
\(
\reach \cap \states_i \neq \emptyset.
\)
\end{definition}

\paragraph{State space segmentation and connectors} 
We examine the possible transitions between the subsets of states for each submodel to determine which submodels are connected. Submodel $\pomdp_i$ connects to submodel $\pomdp_j$ if and only if there exists a state $s \in \states_j$ that is a successor state of $s\in \successor(\states_i)$. We define the connection graph as a directed graph formed by connected submodels. In the presented example, one may move from the top-left quadrant to the top-right or the bottom-left but not the bottom right (see Figure~\ref{fig:Abstraction_Sets_super} for the connection graph of the submodels). We must form this connection graph to give the submodels a $\reach$ specification. Then we form a path from the directed graph created by the connected submodels (dark gray lines in Figure~\ref{fig:Abstraction_Sets_super}).

\paragraph{Compositional reach-avoid procedure}
Consider the non-covered successor states in submodel $\successor(\states_i) \backslash \states_i$, if connected submodel $\pomdp_j$ has a direct path to a winning submodel. Then we label $\reach_i = (\successor(\states_i)\backslash\states_i)\bigcap \states_j \subseteq \hat{\states_i}$.
Any $\avoid \subseteq S$ state will remain an $\avoid_i \subseteq S_i$ in the submodel $\pomdp_i$.  
The reach-avoid specification for submodel $\pomdp_i$ is then given by $\varphi_i = \langle \reach_i,\avoid_i \rangle$.

\paragraph{Reachable spaces}
In example~\ref{ex:grid}, it is insufficient for the reach-avoid shield to synthesize a strategy for the states within the quadrants. For instance, if we consider the set of states inside the dashed lines of the top-left quadrant (covered by $\shield_{\varphi_1}$ in Figure~\ref{fig:Abstraction_Sets}), there is no winning state for the agent to enforce a reach-avoid specification. For the spaces with a subset of states without a winning condition, we modify the winning region to move the agent outside these spaces. This modification includes adding all possible successor states into the set of states covered by the shield.

\paragraph{Initial belief support}
Similarly, if there is no initial state within the set of states, such as the top-right quadrant ($\shield_{\varphi_2}$ in Figure~\ref{fig:Abstraction_Sets}), then there is no input for the synthesis approach to tracking belief support. To initialize such belief support, we overapproximate the input belief support for $\pomdp_2$. We extend the state space of the submodel to include any precursor states for the region (yellow states in figure~\ref{fig:Abstraction_Sets_sub}).
We then include any states from precursor models (backward along the dark gray line in Figure~\ref{fig:Abstraction_Sets_super}) into the initial distribution $I_2(s)$ of the submodel $\pomdp_2$. This modification makes it possible to synthesize a shield from this overapproximation on the belief support.

\begin{algorithm}[!p]
\caption{Compositional Shield Synthesis for POMDPs} \label{alg:synthesis}
\begin{algorithmic}[1]
\Require For POMDP $\pomdp$, a partial model of the structure  $\pomdp' =  (\states,\init,\Act, \ObsFun,\ObsSym, \probmdp')$
\Require Safety specification $\varphi = \langle \text{REACH}, \text{AVOID} \rangle$
\Require Decomposition strategy \Comment{e.g., feature- or location-based}
\Ensure Global shield $\psched$ enforcing $\varphi$ across $\pomdp$

\State \textbf{// Step 1: Decompose Global Model}
\State Construct a set of submodels $\{\pomdp_1', \dots, \pomdp_n'\}$ with core states $S_i$ and extended states $\hat{\states}_i = \states_i \cup \successor(\states_i)$
\For{each submodel $\pomdp_i$}
    \State $\pomdp_i' \gets ( \hat{\states}_i, \init_i, \Act, \ObsFun_i, \ObsSym, \probmdp_i')$
    \State $\probmdp_i(s,a,s)' \gets 1 \quad \forall a \in \Act, \forall s \in \hat{S}_i \setminus S_i$ \Comment{Treat all interface states as absorbing (see Definition~\ref{def:decom})}
\EndFor

\State \textbf{// Step 2: Assign Local Specifications}
\For{each $\pomdp_i'$}
    \State $\avoid_i \gets \avoid \cap \hat{\states}_i$
    \If{$\reach \cap S_i \neq \emptyset$}
        \State $\reach_i \gets \reach \cap \states_i$
    \Else
        \State $\reach_i \gets \successor(\states_i) $ \Comment{use connection graph to assign proxy through successors}
    \EndIf
    \State $\varphi_i \gets \langle \reach_i, \avoid_i \rangle$ \Comment{Define local spec}
\EndFor

\State \textbf{// Step 3: Compute Local Shields}
\For{each $\pomdp_i'$}
    \State $\psched_i \gets \textrm{compute}(\pomdp_i',\varphi_i)$ \Comment{Synthesize local permissive policy (shield) \cite{DBLP:conf/cav/JungesJS20}}
\EndFor

\State \textbf{// Step 4: Compose Global Shield}
\Function{GlobalShield}{$h$}
    \For{each $i \in \{1,\dots,n\}$}
        \State $\belsup_i \gets \estimator_i(h)$ \Comment{Use submodel estimator on current trace}
        \State $\Act_i \gets \psched_i(\belsup_i)$ \Comment{Actions allowed by shield $i$}
    \EndFor
    \State \Return $\psched(h) = \bigcap_{i=1}^n \act_i$ \Comment{Intersection of local shields}
\EndFunction

\State \textbf{// RL Integration}
\State Use $\psched(h)$ to mask actions in RL training (see Figure~\ref{fig:RL_Flowchart})

\end{algorithmic}
\end{algorithm}

\paragraph{Shield synthesis}
We can now compute shields for each submodel. This synthesis procedure can be performed in parallel (Figure~\ref{fig:synthesis_outline} and Algorithm~\ref{alg:synthesis}). For the partial submodel $\pomdp_i'$ and reach-avoid specification $\varphi_i$, we can compute the permissive policy $\psched_i$, which ensures that an agent with belief support $\belsup_i$ will almost-surely satisfy $\varphi_i$.

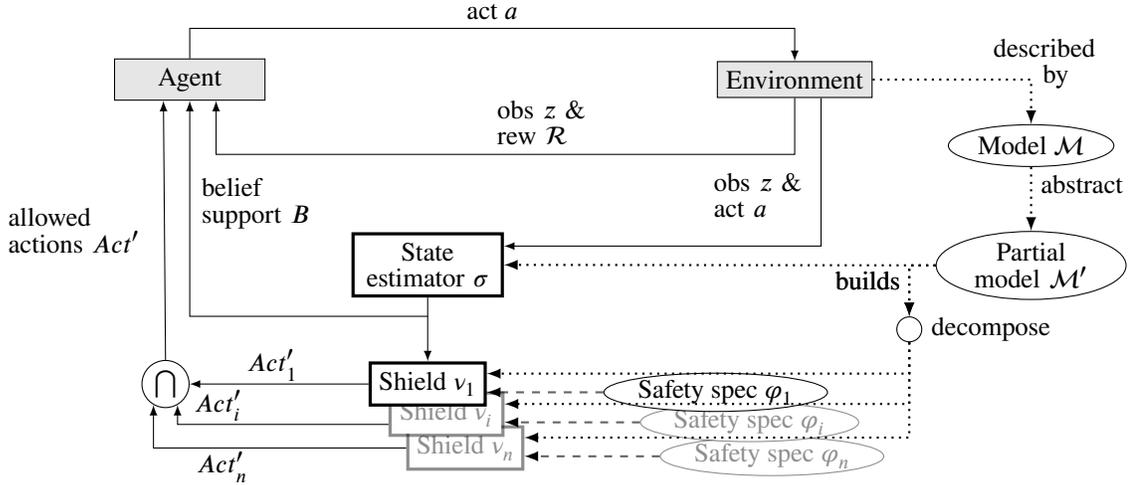
\begin{figure*}[!t]
\hspace{-0.75cm}
\centering
\begin{tikzpicture}
\node[draw, minimum width=2cm,fill=black!10] (agent) { Agent};
\node[right=6cm of agent, draw, minimum width=2cm,fill=black!10] (env) { Environment};

\node[very thick, draw, below=1.75cm of agent,xshift=9em,text width=1.75cm,align=center] (filter) {State \\[-0.2em] estimator $\estimator$};
\node[very thick, draw, below=0.85cm of filter,minimum width=1.5cm,opacity=0.45,xshift=0.5cm,yshift=-0.85cm] (shieldn) {Shield $\psched_n$};
\node[very thick, draw, below=0.85cm of filter,minimum width=1.5cm,text opacity=0.45,xshift=0.25cm,yshift=-0.4cm,fill=white,draw opacity=0.45] (shieldi) {Shield $\psched_i$};
\node[very thick, draw, below=0.85cm of filter,minimum width=1.5cm,fill=white] (shield) {Shield $\psched_1$};

\draw[->] (agent.north) |- ($(agent.north)+(0,0.4)$)  -| node[above,pos=0.25] {act $\act$} (env.north);
\draw[->] (env.south) |- ($(env.south)+(0,-0.75)$)  -| node[pos=0.20,above=-0.0cm,text width=1.75cm] {obs $\obs$ \& \\[-0.2em] rew $\rew$} ($(agent.south)+(0.35,0)$);

\node[right=1cm of env.0,ellipse,draw,yshift=-2.5em,inner sep=2.0pt] (model) { Model $\pomdp$};

\node[ellipse,below=.85cm of model,draw,text width=1.75cm,align=center,inner sep=0.5pt,minimum width=2.5cm] (graph) { Partial\\[-0.2em] model $\pomdp'$};

\draw[->,dotted, thick] (env.0) -| node[above=-.25cm,text width=1.0cm,align = right] { described \\[-0.2em] by} (model.north);
\draw[->,dotted, thick] (model.south) -| node[right=0.0cm,pos=0.5,text width=2.5cm,yshift=-0.65em] {abstract} (graph.north);

\node[ellipse,right=1.55cm of shieldn.352,draw,inner sep=0.5pt,text opacity=0.45,xshift=0.25cm,fill=white,draw opacity=0.45] (specn) {Safety spec $\varphi_n$};

\node[ellipse,right=1.55cm of shieldi.352,draw,inner sep=0.5pt,text opacity=0.45,xshift=0.25cm,fill=white,draw opacity=0.45] (speci) {Safety spec $\varphi_i$};

\node[ellipse,right=1.55cm of shield.352,draw,inner sep=0.5pt,fill=white] (spec) {Safety spec $\varphi_1$};

\draw[->] ($(filter.south) + (0,-0.25)$)  -| node[pos=0.25,left,text width=1.5cm,xshift=0.2cm,yshift=1.5cm,font=] {belief\\[-0.2em] support $\belsup$} ($(agent.south)+(00.0,0)$);
\draw[thick, dashed,->,draw opacity=0.6] (spec) -- node[below,pos=0.27] {} (shield.352);

\draw[thick, dashed,->,draw opacity=0.6] (speci) -- node[below,pos=0.27] {} (shieldi.352);
\draw[thick, dashed,->,draw opacity=0.6] (specn) -- node[below,pos=0.27] {} (shieldn.352);

\draw[thick, dotted,->] (graph) -- node[pos=0.5,below]{} ($(filter.0)+(0.0,0.0)$);

\draw[->] ($(env.south)+(0.35,0)$) |- node [right=0.05cm,pos=0.1,text width=1.5cm,yshift=-0.90cm,xshift=-1.6cm] {obs $\obs$  \& \\[-0.2em] act $\act$}($(filter.east)+(0.0,0.25)$);
\draw[->] (filter) -- (shield);

\node[circle,draw,minimum width=0.20cm,label=right:decompose] (decompose) at ($(graph.west) + (-0.35,-0.85)$){};
\draw[thick, dotted,->] ($(graph.west)$) -| node[xshift=-0.55cm,yshift=-0.20cm]{builds}  (decompose.90);

\draw[thick, dotted,->] ($(graph.west)$) -| node[xshift=-0.55cm,yshift=-0.20cm]{builds}  (decompose.90);
\draw[thick, dotted,->] (decompose.270) |- node[xshift=-0.47cm,yshift=-0.15cm]{}  (shield.10);
\draw[thick, dotted,->] (decompose.270) |- node[xshift=-0.47cm,yshift=-0.15cm]{}  (shieldi.10);
\draw[thick, dotted,->] (decompose.270) |- node[xshift=-0.47cm,yshift=-0.15cm]{}  (shieldn.10);

\node[circle,draw,minimum width=0.10cm,inner sep=0.05cm] (actions) at ($(shield.west) + (-2.7,0)$){$\bigcap$};

\draw[->] (actions.90) -- node[text width=1.85cm,xshift=-1.15cm,yshift=0.0cm]{allowed \\[-0.2em] actions~$\Act'$} ($(agent.south)-(0.35,0)$);
\draw[->] (shield.180) -- node[text width=1.85cm,xshift=0.5cm,yshift=0.25cm]{$\Act'_1$} (actions.0);
\draw[->] (shieldi.190) -| node[text width=1.85cm,xshift=1.15cm,yshift=0.25cm]{$\Act'_i$} (actions.300);
\draw[->] (shieldn.180) -| node[text width=1.85cm,xshift=1.5cm,yshift=-0.25cm]{$\Act'_n$} (actions.240);


\end{tikzpicture}
	\caption{RL with a compositional shield. %
 }
	\label{fig:RL_Flowchart}
\end{figure*}

\section{SAFE RL VIA~COMPOSITIONAL SHIELDS}

We present the problem of safe learning in POMDPs.

\begin{problem}[Safety-critical RL]\label{problem:max-rew-reach-avoid}
	Given a POMDP $\pomdp$, a safety constraint $\varphi$, and let $\osched_1,\ldots,\osched_n$ be the  (training) sequence of policies employed by an agent. 
The problem is to ensure that for all policies $\osched_i$ it holds that $\pomdp(\osched_i)\models\varphi$ with $1\leq i\leq n$ and the final policy $\pi_n$ maximizes~$\mathbb{E}\left[\sum_{t=0}^\infty\gamma^t\rew_t\right]$.
\end{problem}

An optimal solution to problem~\ref{problem:max-rew-reach-avoid} may induce a lower reward than unconstrained learning, as the agent has to strictly adhere to the safety constraint while collecting rewards.

A shielded agent augments the reinforcement learning loop~\cite{sutton1998reinforcement} (Figure~\ref{fig:RL_Flowchart}).
In particular, at each decision step, the RL agent selects an action $a$ from the set of available actions $\Act(s)$. The shielded RL agent selects from the subset of these actions $\Act'(s)$ created by the permissive policy $\psched$.
Additionally, the RL agent uses the graph knowledge to track the belief support.

\subsection{Shields and RL in POMDPs}

Enforcing shielding for safe RL is useful for environments with sparse reward and incomplete information~\cite{DBLP:conf/aaai/Carr0JT23}. However, one limitation is that the size of the potential belief supports $\belsup$ increases exponentially with the number of aliased observations.

In this work, we formulate a procedure that allows an RL agent to compute a shield by parts from a set of connected sub-shields with two benefits. First, the sub-shield upper bounds the size of belief support, ensuring a computable and tractable shield. Second, it provides the learner with additional information about the features of the environment as it learns, which allows for warm-start learning~\cite{neary2022verifiable} in sparse reward settings.

\subsection{Shield Properties}
First we describe the important properties that are preserved when deploying shielding in deep RL. In particular, we motivate the use of a partial model, which representations an overapproximation of the POMDP graph.

\paragraph{Using the partial model}
The \textit{belief-support state estimator} $\estimator$  provides a light-weight, interface to track the possible system states at every loop of Figure~\ref{fig:RL_Flowchart}. This estimator also provides the input for the shield (definition~\ref{def:shield}).

\paragraph{Using the partial model via a shield}
We assume the availability of a shield that ensures reach-avoid specifications as outlined above (definition~\ref{def:shield}).
Using the state estimator, we can synthesize a shield $\psched\colon \supp(\belief)\rightarrow 2^\Act$  that operates directly on the belief support. For the specification $\varphi$, this shield yields for every belief the set of \emph{safe actions}. 
We restrict the available actions for the agent to these safe actions. For the specification $\varphi$, this shield yields for every belief the set of \emph{safe actions}.  We restrict the available actions for the agent to these safe actions.

\vspace{1em}
\subsection{Learning Outcomes}
Shielded RL agents can enforce safe decisions during learning. Without a shield, the agent must take action to record that it may transition to an unsafe state. While a shield ensures that the agent visits the reach states eventually (with probability one), there is no upper bound on the number of steps required to visit these states. The shielded agent may violate the specification if the partial model is not faithful -- it has a different graph structure -- to the true POMDP. Moreover, state estimators themselves do not directly enforce safe exploration. However, precise tracking of the belief support augments the learning agent with structured information.

\paragraph{Safety after learning}
Even after the RL agent has finished learning, agents who keep exploring may violate safety. Furthermore, reward objectives and safety constraints may not coincide, e.g., a car driving to a location. A reward structure may only focus on a particular feature, such as arriving with enough fuel, while a safety constraint only enforces collision avoidance. In such cases, there is a risk-reward consideration where the RL may accept a small but non-zero risk to avoid a significant and expensive re-route.

\paragraph{RL convergence speed}
Learning in partially observable settings remains challenging beyond providing safety guarantees, especially when rewards are sparse. The availability of a partial model can accelerate the learning process. In particular, the availability of a state estimator allows enriching the observation with a signal that compresses the history and the shield induces longer episodes that improve learning in sparse-reward environments ~\cite{DBLP:conf/aaai/Carr0JT23}.

\subsection{Compositional Shield Execution}
After computing a shield, we can filter and apply the belief support relevant to the states covered in its submodel (dashed segment in Figure~\ref{fig:synthesis_outline}). For example, if the agent's belief support contains two states: one in the top-left quadrant and one in the top-right quadrant, then the input belief support $B$ is filtered to a single state for shield $\psched_1$ and a single state for $\psched_2$.
Finally, we constrain the RL agent with the intersect of the sets of allowed actions $\Act' = \bigcap_{i=1}^{n}  \Act'_i$, eliminating any possible violating actions for all states in the belief support $\belsup$.

\subsection{Parallel Training}
One advantage of the decomposition and its corresponding shields is that we can use a learning method that is compatible and operable with compositional RL frameworks~\cite{neary2022verifiable}.
We may deploy each RL subsystem and subtask, initializing a subtask learning problem for each subtask and submodel, thereby emphasizing relevant regions and ignoring low-value subsystems.
While a deep exploration of the various compositional RL frameworks is beyond the scope of this paper, we demonstrate that this approach to shielding is compatible with higher abstraction levels.

\begin{table}[!t]

    \caption{Individual sizes and total compute times for shield and sub-shield synthesis.}
    \label{tab:synthesis}
    
    \renewcommand{\arraystretch}{1.4}
    \centering
    \begin{tabular}{lrrrr}
		\hline
		Domain (param) & \multicolumn{2}{c}{Shield Synthesis} & \multicolumn{2}{c}{Compositional} \\
		&  \multicolumn{1}{r}{\emph{States}} & \multicolumn{1}{r}{\emph{Time (s)}} 	&  \multicolumn{1}{r}{\emph{States}} & \multicolumn{1}{r}{\emph{Time (s)}} \\
		\hline
		\domain{Obstacle}(6) & 37 & 1 & 26 & 2 \\
		\domain{Obstacle}(8) & 65  & 18 & 37 & 5  \\
		\domain{Obstacle}(10) & 101 & 598 & 50 & 55   \\
		\domain{Obstacle}(16) &  257 & 6314 & 101 & 1504 \\
		\domain{Obstacle}(20) & 401 & \TO &  145 & 1605\\
		\rowcolor{gray!25} \domain{Refuel}(6,8) & 270 & 6 & 101 & 6 \\
		\rowcolor{gray!25} \domain{Refuel}(8,10) & 641  & 109 & 270 & 24  \\
		\rowcolor{gray!25} \domain{Refuel}(10,12) & 1201 & 816  & 651 & 150  \\
		\rowcolor{gray!25} \domain{Refuel}(16,18) &  4609 & \TO & 1235 & 860  \\
		\domain{Evade}(6,2) & 4232 & 142 & 1833 & 120  \\
		\domain{Evade}(8,2) & 10368 & 1510 & 2501  & 213\\
		\domain{Evade}(10,2) & 16391 & \TO & 7121 & 1609 \\
		\domain{Evade}(16,2) & \TO & \TO & 14501 & 9601 \\
		\rowcolor{gray!25} \domain{Intercept}(7,1) & 4705 & 116 & 2501 & 32 \\
		\rowcolor{gray!25} \domain{Intercept}(10,1) & 9690 & 1861 & 4705 & 450 \\
		\rowcolor{gray!25} \domain{Intercept}(16,1) & 49235  & \TO & 10032 & 4601\\
		\hline
	\end{tabular}
\end{table}

\section{EMPIRICAL ANALYSIS}

We evaluate the effectiveness of compositionally shielded reinforcement learning (RL) in four partially observable domains with sparse reward functions. These domains, adapted from \cite{DBLP:conf/cav/JungesJS20} and \cite{DBLP:conf/aaai/Carr0JT23}, present unique challenges for RL agents:
\begin{enumerate}
\item \domain{Obstacle}: An agent navigates a maze (movement reward: $-1$) with static traps, uncertain initial state, and noisy movement distance. The agent only observes whether its current position is a trap ($-10^3$ reward) or the exit ($+10^3$ reward).

\item \domain{Refuel}: A rover traverses from one corner to another ($+10$ reward), avoiding a diagonal obstacle. Movement consumes energy, and the rover can recharge at designated stations to full capacity without rewards or costs. The rover has noisy sensors for position and battery level. Collisions and battery depletion terminate the episode.

\item \domain{Evade}: An agent must reach an escape door ($+10$ reward) while evading a faster robot. The agent has a limited vision range ($\textrm{Radius}$), but can scan the entire grid instead of moving.

\item \domain{Intercept}: Compared to \domain{Evade}, the agent aims to intercept ($+10^3$ reward) a robot before it exits the grid ($-10^3$ reward). The agent has a view radius and observes a central corridor. Movements incur a $-1$ reward.
\end{enumerate}

We employ five deep RL methods: DQN \cite{DBLP:journals/nature/MnihKSRVBGRFOPB15}, DDQN \cite{DBLP:conf/aaai/HasseltGS16}, PPO~\cite{DBLP:journals/corr/SchulmanWDRK17}, discrete SAC \cite{DBLP:journals/corr/abs-1910-07207}, and REINFORCE \cite{DBLP:journals/ml/Williams92}. Episodes are limited to 100 steps, and average rewards are calculated across 10 evaluation episodes. To improve readability, we apply smoothing to all figures using a five-interval window.

For both the complete model and the submodels, we use the \emph{Storm} framework \cite{DBLP:journals/sttt/HenselJKQV22} to interface the model, shield, and state estimator. We use the same tool for both approaches to the problem (composition vs. no composition) for a fair comparison with the state-of-the-art.
Bindings to Tensorflow's \emph{TF-Agents} package \cite{TFAgents} are deployed, and its masking function is used to implement the precomputed shield \cite{DBLP:conf/aaai/Carr0JT23}. Experiments are conducted on an $8 \times 3.2$~GHz Intel Xeon Platinum 8000 series processor with 32 GB of RAM. Computations exceeding $10^5$ seconds are marked as timed out~(TO).

\subsection{Shield Synthesis through Subproblem Composition}

\paragraph{Domain decompositions}
We employ quadrant-based decompositions for the \domain{Obstacle} and \domain{Refuel} domains. For \domain{Evade} and \domain{Intercept}, we use a dual quadrant-based decomposition (16 submodels) to represent agent and adversary locations in each region.

\paragraph{Scalability improvements}
Traditional shield synthesis \cite{DBLP:conf/cav/JungesJS20,DBLP:conf/aaai/Carr0JT23} often fails to scale beyond $10^4$ states (Table~\ref{tab:synthesis}). Our compositional approach computes shields for these and larger domains, achieving a two-order-of-magnitude scalability improvement over centralized methods.

\begin{figure}[!t]
	\centering
 \begin{adjustbox}{width=.65\columnwidth,trim={1cm 0pt 0pt 0pt}, clip}
	\pgfplotstableread{
0.000000000000000000e+00 0.28
1.000000000000000000e+02 0.38
2.000000000000000000e+02 0.35
3.000000000000000000e+02 0.39
4.000000000000000000e+02 0.35
5.000000000000000000e+02 0.35
6.000000000000000000e+02 0.38
7.000000000000000000e+02 0.38
8.000000000000000000e+02 0.40
9.000000000000000000e+02 0.43
1.000000000000000000e+03 0.45
1.100000000000000000e+03 0.47
1.200000000000000000e+03 0.48
1.300000000000000000e+03 0.49
1.400000000000000000e+03 0.49
1.500000000000000000e+03 0.50
1.600000000000000000e+03 0.50
1.700000000000000000e+03 0.55
1.800000000000000000e+03 0.59
1.900000000000000000e+03 0.60
2.000000000000000000e+03 0.61
2.100000000000000000e+03 0.61
2.200000000000000000e+03 0.59
2.300000000000000000e+03 0.57
2.400000000000000000e+03 0.58
2.500000000000000000e+03 0.56
2.600000000000000000e+03 0.57
2.700000000000000000e+03 0.57
2.800000000000000000e+03 0.58
2.900000000000000000e+03 0.55
3.000000000000000000e+03 0.56
3.100000000000000000e+03 0.56
3.200000000000000000e+03 0.58
3.300000000000000000e+03 0.58
3.400000000000000000e+03 0.61
3.500000000000000000e+03 0.61
3.600000000000000000e+03 0.62
3.700000000000000000e+03 0.61
3.800000000000000000e+03 0.61
3.900000000000000000e+03 0.60
4.000000000000000000e+03 0.61
4.100000000000000000e+03 0.57
4.200000000000000000e+03 0.56
4.300000000000000000e+03 0.55
4.400000000000000000e+03 0.57
4.500000000000000000e+03 0.57
4.600000000000000000e+03 0.59
4.700000000000000000e+03 0.61
4.800000000000000000e+03 0.64
4.900000000000000000e+03 0.62
5.000000000000000000e+03 0.64
}\AblationNormNoShield
\pgfplotstableread{
0.000000000000000000e+00 4.447427795992957011e-01
1.000000000000000000e+02 7.393847317165799238e-01
2.000000000000000000e+02 7.561213921440971264e-01
3.000000000000000000e+02 7.859813916948105161e-01
4.000000000000000000e+02 7.773361059824624997e-01
5.000000000000000000e+02 7.007132802274491645e-01
6.000000000000000000e+02 7.646910032908120769e-01
7.000000000000000000e+02 7.683866658104790970e-01
8.000000000000000000e+02 7.948928319719102964e-01
9.000000000000000000e+02 8.088634432050917367e-01
1.000000000000000000e+03 8.232137785169815203e-01
1.100000000000000000e+03 8.464571091545952664e-01
1.200000000000000000e+03 8.683430542416044062e-01
1.300000000000000000e+03 8.681819435967339516e-01
1.400000000000000000e+03 8.809951086256239128e-01
1.500000000000000000e+03 8.815550535837809454e-01
1.600000000000000000e+03 8.904852762858072035e-01
1.700000000000000000e+03 8.965556677500406968e-01
1.800000000000000000e+03 9.021679441240100239e-01
1.900000000000000000e+03 9.052623887803820102e-01
2.000000000000000000e+03 9.185226098378498349e-01
2.100000000000000000e+03 9.174962216271294047e-01
2.200000000000000000e+03 9.223634423573812136e-01
2.300000000000000000e+03 9.178537777370876505e-01
2.400000000000000000e+03 9.089460575527614772e-01
2.500000000000000000e+03 9.050961136288112607e-01
2.600000000000000000e+03 9.038247255537243952e-01
2.700000000000000000e+03 8.982046154446071684e-01
2.800000000000000000e+03 9.020177231258815764e-01
2.900000000000000000e+03 9.058806101904974817e-01
3.000000000000000000e+03 9.113993868933781828e-01
3.100000000000000000e+03 9.142522193060980129e-01
3.200000000000000000e+03 9.184163290659587142e-01
3.300000000000000000e+03 9.183673329671223629e-01
3.400000000000000000e+03 9.187803879208034097e-01
3.500000000000000000e+03 9.140100548638238065e-01
3.600000000000000000e+03 9.092947213914660542e-01
3.700000000000000000e+03 9.095436663733589278e-01
3.800000000000000000e+03 9.087982211642794406e-01
3.900000000000000000e+03 9.054636096530490619e-01
4.000000000000000000e+03 9.081865538703070007e-01
4.100000000000000000e+03 9.114418324788410297e-01
4.200000000000000000e+03 9.119296665615506114e-01
4.300000000000000000e+03 9.133271090189615249e-01
4.400000000000000000e+03 9.177191109975177952e-01
4.500000000000000000e+03 9.118659447140163854e-01
4.600000000000000000e+03 9.135789441850450388e-01
4.700000000000000000e+03 9.123509997897678581e-01
4.800000000000000000e+03 9.096843341403536565e-01
4.900000000000000000e+03 9.083950557284884431e-01
5.000000000000000000e+03 9.152004997253418983e-01
}\AblationNormShield
\pgfplotstableread{
0.000000000000000000e+00 3.447427795992957011e-01
1.000000000000000000e+02 3.814335624557426238e-01
2.000000000000000000e+02 4.561213921440971264e-01
3.000000000000000000e+02 5.159813916948105161e-01
4.000000000000000000e+02 4.773361059824624997e-01
5.000000000000000000e+02 5.007132802274491645e-01
6.000000000000000000e+02 6.646910032908120769e-01
7.000000000000000000e+02 5.683866658104790970e-01
8.000000000000000000e+02 6.048928319719102964e-01
9.000000000000000000e+02 7.088634432050917367e-01
1.000000000000000000e+03 7.932137785169815203e-01
1.100000000000000000e+03 7.364571091545952664e-01
1.200000000000000000e+03 7.683430542416044062e-01
1.300000000000000000e+03 7.681819435967339516e-01
1.400000000000000000e+03 7.809951086256239128e-01
1.500000000000000000e+03 7.815550535837809454e-01
1.600000000000000000e+03 7.304852762858072035e-01
1.700000000000000000e+03 8.465556677500406968e-01
1.800000000000000000e+03 8.092309104398841857e-01
1.900000000000000000e+03 8.651823018237234621e-01
2.000000000000000000e+03 8.981283174319759210e-01
2.100000000000000000e+03 8.474962216271294047e-01
2.200000000000000000e+03 8.932821381249814711e-01
2.300000000000000000e+03 9.079230193081417471e-01
2.400000000000000000e+03 9.142123421747649894e-01
2.500000000000000000e+03 9.130941085178247172e-01
2.600000000000000000e+03 9.132193123234587375e-01
2.700000000000000000e+03 9.102046154446071684e-01
2.800000000000000000e+03 8.735677231258815764e-01
2.900000000000000000e+03 8.85678899904974817e-01
3.000000000000000000e+03 8.703993868933781828e-01
3.100000000000000000e+03 8.9022522193060980129e-01
3.200000000000000000e+03 9.158806101904974817e-01
3.300000000000000000e+03 9.083673329671223629e-01
3.400000000000000000e+03 9.077803879208034097e-01
3.500000000000000000e+03 9.040100548638238065e-01
3.600000000000000000e+03 9.192947213914660542e-01
3.700000000000000000e+03 9.195436663733589278e-01
3.800000000000000000e+03 9.187982211642794406e-01
3.900000000000000000e+03 9.154636096530490619e-01
4.000000000000000000e+03 9.181865538703070007e-01
4.100000000000000000e+03 9.052004997253418983e-01
4.200000000000000000e+03 9.219296665615506114e-01
4.300000000000000000e+03 9.033271090189615249e-01
4.400000000000000000e+03 9.177191109975177952e-01
4.500000000000000000e+03 9.218659447140163854e-01
4.600000000000000000e+03 9.078537777370876505e-01
4.700000000000000000e+03 9.192994184529819357e-01
4.800000000000000000e+03 9.196843341403536565e-01
4.900000000000000000e+03 9.183950557284884431e-01
5.000000000000000000e+03 9.223214795165219903e-01
}\AblationNormShieldComposed
\definecolor{brightgreen}{rgb}{0.616, 0.847, 0.40}
\definecolor{amethyst}{rgb}{0.6, 0.4, 0.8}
\definecolor{orangehue}{rgb}{1.0,0.627,0.337}
\definecolor{brightblue}{rgb}{0.553, 0.867, 0.8167}

\begin{tikzpicture}
\begin{axis}[name=rl_avoid,
    axis lines*=left, 
    xtick = {5000},
    enlarge y limits=false,
    enlarge x limits=false, 
    tick style={black},
xlabel = {Number of episodes},
ylabel = {Normalized reward},
ylabel style={yshift=0.2cm},
xmin=0,
xmax=5000,
ymin= 0,
ymax= 1,
width = 0.6\textwidth,
height = 0.6\textwidth,
xticklabel style={/pgf/number format/precision=2},
no markers,
every axis legend/.code={\let\addlegendentry\relax}
]

\addplot[lightgray,very thick] table[x index={0}, y index={1},col sep=space] from \AblationNormNoShield;
\addlegendentry{No shield}

\addplot[gray,very thick] table[x index={0}, y index={1},col sep=space] from \AblationNormShield;
\addlegendentry{Shielded}

\addplot[Goldenrod2,very thick] table[x index={0}, y index={1},col sep=space] from \AblationNormShieldComposed;
\addlegendentry{Compositional shield}    

\node at (axis cs:2500,.96) [anchor=west] {\textcolor{Goldenrod2}{compositional shield}};
\node at (axis cs:3100,.86) [anchor=west] {\textcolor{gray}{shield}};
\node at (axis cs:3050,.54) [anchor=west] {\textcolor{lightgray}{no shield}};

\end{axis}
\end{tikzpicture}
 \end{adjustbox}
	\caption{Learning agents with different shield types on a reward normalized across the set of small domains.}
	\label{fig:empirical_small}
\end{figure}
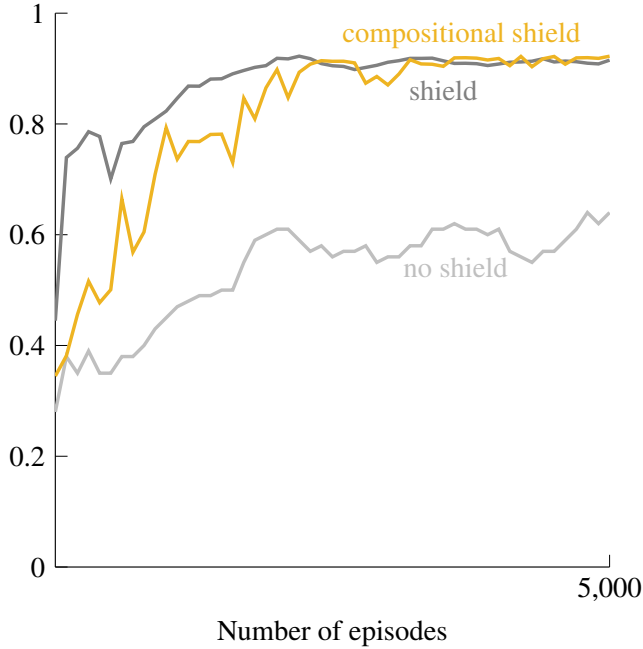

\paragraph{Learning performance}
We normalize RL agent performance to [0, 1], with 1 representing the highest achievable value. In small domains, compositional shields require slightly more data to converge than centralized shields (Figure~\ref{fig:empirical_small}). However, the compositional RL framework improves data efficiency, with faster convergence than initializing RL only in initial states (Figure~\ref{fig:emperical_large}).

\paragraph{Safety and performance in large domains}
In large domains where centralized shield synthesis is infeasible, agents with compositional shields significantly outperform unshielded agents (Figure~\ref{fig:emperical_large}). Shielding ensures safety throughout learning, with sub-shields maintaining this property (Table~\ref{tab:shield_metrics}). Although sub-shields may converge slower in small domains, compositional shielding is crucial for safe and efficient learning in complex environments.

For ease of comparison across domains and learning methods, we normalized the RL agent performance across all domains. We modify the reward such that the agent will take a value between $0$ and $1$, where $1$ is the known highest possible value that the RL agent can achieve.

\paragraph{Equivalent performance to centralized shield in small domains} When comparing the learning performance for the shield with the compositional shield, we observe that the compositional shield needs slightly more data to converge (dark gray line compared to the yellow line in Figure~\ref{fig:empirical_small}).

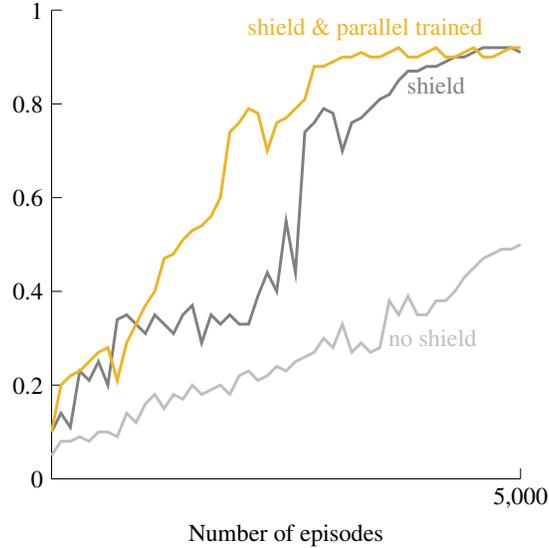
\begin{figure}[!t]
	\centering
  \begin{adjustbox}{width=.55\columnwidth,trim={1cm 0pt 0pt 0pt}, clip}
	\pgfplotstableread{
0	0.10
100	0.14
200	0.11
300	0.23
400	0.21
500	0.25
600	0.20
700	0.34
800	0.35
900	0.33
1000 0.31
1100 0.35
1200 0.33
1300 0.31
1400 0.35
1500 0.37
1600 0.29
1700 0.35
1800 0.33
1900 0.35
2000 0.33
2100 0.33
2200 0.39
2300 0.44
2400 0.40
2500 0.55
2600 0.44
2700 0.74
2800 0.76
2900 0.79
3000 0.78
3100 0.70
3200 0.76
3300 0.77
3400 0.79
3500 0.81
3600 0.82
3700 0.85
3800 0.87
3900 0.87
4000 0.88
4100 0.88
4200 0.89
4300 0.90
4400 0.90
4500 0.91
4600 0.92
4700 0.92
4800 0.92
4900 0.92
5000 0.91
}\AblationNormShieldLarge
\pgfplotstableread{
0	0.05
100	0.08
200	0.08
300	0.09
400	0.08
500	0.10
600	0.10
700	0.09
800	0.14
900	0.12
1000 0.16
1100	0.18
1200	0.15
1300	0.18
1400	0.17
1500	0.20
1600	0.18
1700	0.19
1800	0.20
1900	0.18
2000	0.22
2100	0.23
2200	0.21
2300	0.22
2400	0.24
2500	0.23
2600	0.25
2700	0.26
2800	0.27
2900	0.30
3000	0.28
3100	0.33
3200	0.27
3300	0.29
3400	0.27
3500	0.28
3600	0.38
3700	0.35
3800	0.39
3900	0.35
4000	0.35
4100	0.38
4200	0.38
4300	0.40
4400	0.43
4500	0.45
4600	0.47
4700	0.48
4800	0.49
4900	0.49
5000	0.50
}\AblationNormNoShieldLarge
\pgfplotstableread{
0	0.1
100	0.2
200	0.22
300	0.23
400	0.25
500	0.27
600	0.28
700	0.21
800	0.29
900	0.33
1000 0.37
1100	0.4
1200	0.47
1300	0.48
1400	0.51
1500	0.53
1600	0.54
1700	0.56
1800	0.6
1900	0.74
2000	0.76
2100	0.79
2200	0.78
2300	0.70
2400	0.76
2500	0.77
2600	0.79
2700	0.81
2800	0.88
2900	0.88
3000	0.89
3100	0.90
3200	0.90
3300	0.91
3400	0.90
3500	0.90
3600	0.91
3700	0.92
3800	0.90
3900	0.90
4000	0.91
4100	0.92
4200	0.90
4300	0.90
4400	0.91
4500	0.92
4600	0.90
4700	0.90
4800	0.91
4900	0.92
5000	0.92
}\AblationNormShieldLargeParallel
\definecolor{brightgreen}{rgb}{0.616, 0.847, 0.40}
\definecolor{amethyst}{rgb}{0.6, 0.4, 0.8}
\definecolor{orangehue}{rgb}{1.0,0.627,0.337}
\definecolor{brightblue}{rgb}{0.553, 0.867, 0.8167}

\begin{tikzpicture}
	\begin{axis}[name=rl_avoid,
    axis lines*=left, 
    xtick = {5000},
    enlarge y limits=false,
    enlarge x limits=false, 
    tick style={black},
	xlabel = {Number of episodes},
	ylabel = {Normalized reward},
	ylabel style={yshift=0.2cm},
	xmin=0,
	xmax=5000,
	ymin= 0,
	ymax= 1,
	width = .6\textwidth,
	height =.6\textwidth,
	no markers,
    every axis legend/.code={\let\addlegendentry\relax}]

\addplot[lightgray,very thick] table[x index={0}, y index={1},col sep=space] from \AblationNormNoShieldLarge;
     \addlegendentry{No shield}

\addplot[gray,very thick] table[x index={0}, y index={1},col sep=space] from \AblationNormShieldLarge;
\addlegendentry{Shielded}

\addplot[Goldenrod2,very thick] table[x index={0}, y index={1},col sep=space] from \AblationNormShieldLargeParallel;
     \addlegendentry{Shielded \& parallel trained }    
\node at (axis cs:2000,.96) [anchor=west] {\textcolor{Goldenrod2}{shield \& parallel trained}};
\node at (axis cs:3700,.84) [anchor=west] {\textcolor{gray}{shield}};
\node at (axis cs:3500,.3) [anchor=west] {\textcolor{lightgray}{no shield}};

\end{axis}
\end{tikzpicture}
  \end{adjustbox}
	\caption{Learning agents with different shield types on a reward normalized across the set of domains too large for a centralized shield.}
	\label{fig:emperical_large}
\end{figure}

\paragraph{Improved performance compared to unshielded agents in larger domains} When we compare the shielded agent to that of an unshielded agent in very large domains that we are unable to synthesize shields for (dark gray outperforms light gray in Figure~\ref{fig:emperical_large}).

\paragraph{Improved data efficiency using compositional RL} When we deploy a compositional RL framework, the RL converges faster than a simple problem where we perform RL initialized only in initial states (yellow line in Figure~\ref{fig:emperical_large} converges around $2000$ episode while the dark gray line converges around $3500$). Here, the RL agent may initialize the belief support on the sub-models that are not inside the set of initial distribution, which allows the RL to better explore regions of interest without relying on repetitive trajectories.

\paragraph{Enforcing safety while learning} Shielding ensures that the agent is safe while learning, i.e., it will never violate the safety specification $\varphi$ (table~\ref{tab:shield_metrics}). We demonstrate that the sub-shields maintain this property but it does not as easily enforce that the RL agent will reach the goal during the learning process. This condition leads to the sub-shields taking slightly longer to converge in smaller domains when compared to overall shield. However, the RL agent employing a compositional shield still significantly outperforms an unshielded agent in larger domains.

\begin{table}[!t]
\renewcommand{\arraystretch}{1.3}
    \centering
    \caption{Average violation and positive rewards during and after reinforcement learning.}
    \label{tab:shield_metrics}
	\begin{tabular}{lrrrr}
		\hline
		Learning Setting & \multicolumn{2}{c}{No. violations} & \multicolumn{2}{c}{\% successful runs}\\
		&  \multicolumn{1}{r}{\emph{During}} & \multicolumn{1}{r}{\emph{After}} 	&  \multicolumn{1}{r}{\emph{During}} & \multicolumn{1}{r}{\emph{After}}  \\
		\hline
		\emph{No shield} & 3153 & 1023 & 25\% & 40\% \\
		\emph{Shield} & 0  & 0  & 90\% & 95\%  \\
		\emph{Sub-shields} & 0 & 0 & 75\% & 95\% \\
		\hline
	\end{tabular}
\end{table}

\paragraph{Automated Approaches to Decomposition}
While the decompositions used in our experiments (e.g., quadrant-based or feature-partitioned) are manually chosen to match domain symmetries, we acknowledge the need for more systematic methods. One potential direction is to leverage the sparsity structure of the transition graph associated with the POMDP $\pomdp$. 
Specifically, we can treat the POMDP’s transition kernel $\probmdp'$ as a sparse adjacency matrix and apply graph partitioning or sparse matrix decomposition techniques such as spectral clustering \cite{azizzadenesheli2016reinforcement} to identify weakly connected subgraphs. 
These subgraphs correspond to regions of the state space with limited interconnectivity, making them ideal candidates for submodels. Such an approach not only reduces cross-submodel dependencies (and hence conservatism in the shield) but also provides a data-driven, automated method for selecting a “good” decomposition that scales with the topology of the environment. We consider the design and analysis of such structured decomposition strategies a valuable direction for future work.

\section{CONCLUSION}

Shield synthesis is an effective method for ensuring safety requirements in systems that operate within an intrinsic uncertainty. However, shielding needs to scale more efficiently with the size of the problem formulation. Exploiting compositionality in the POMDP problem formulation parallelizes shield synthesis, making it more efficient without sacrificing effectiveness. In particular, compositional shield synthesis relies on connecting state and action spaces between submodel formations via information sharing. We demonstrate this increase in efficiency and preservation of shield properties using compositional synthesis with state-of-the-art deep RL algorithms in sparse-reward and partially observable environments.

\section*{ACKNOWLEDGMENTS}
S. Carr and U. Topcu are partially supported by the US DoD Air Force Office of Scientific Research under award number AFOSR FA9550-22-1-0403 and the  National Science Foundation (NSF) under award number NSF 2409535. 

G. Bakirtzis is partially supported by the academic and research chair \emph{Architecture des Systemes Complexes} through the following parters: Dassault Aviation, Naval Group, Dassasut Systemes, KNDS France, Agence de L'Innovation de Defense and Institut Polytechnique de Paris.

The authors would like to thank Sebastian Junges and Nils Jansen for input
at the early stages of this research.
\bibliography{manuscript}
\bibliographystyle{IEEEtran}
\end{document}